\documentclass[11pt]{article}

\usepackage[T1]{fontenc}
\usepackage[margin = 2.5 cm]{geometry}
\usepackage{enumitem}
\usepackage{float}
\usepackage{numbertabbing}
\usepackage{url}
\usepackage[hidelinks]{hyperref}
\usepackage{xspace}
\usepackage{amsthm}
\usepackage{amsmath} 
\usepackage{amssymb} 
\usepackage{subcaption}
\usepackage[dvipsnames]{xcolor} 
\usepackage[disable]{todonotes}
\usepackage{amstext,amsfonts,bm}
\usepackage{thm-restate}
\usepackage[parfill]{parskip}
\usepackage{authblk}
\usepackage{graphicx}
\usepackage[backend=bibtex, style = numeric-comp]{biblatex}
\AtEveryBibitem{\clearfield{doi}}

\RequirePackage[T1]{fontenc} \RequirePackage[tt=false, type1=true]{libertine} \RequirePackage[varqu]{zi4} \RequirePackage[libertine]{newtxmath}

\floatstyle{ruled}
\newfloat{algo}{htbp}{algo}
\floatname{algo}{Algorithm}

\def \ifempty#1{\def\temp{#1} \ifx\temp\empty }
\renewcommand{\P}{\mathbb{P}}
\newcommand{\E}{\mathbb{E}}
\newcommand{\str}[1]{\textsc{#1}}
\newcommand{\var}[1]{\textit{#1}}
\newcommand{\op}[1]{\textsl{#1}}
\newcommand{\msg}[2]{\ensuremath{\ifempty{#2} [\str{#1}] \else [\str{#1}, {#2}] \fi}}

\newcommand{\false}{\textsc{false}\xspace}
\newcommand{\true}{\textsc{true}\xspace}

\newcommand{\CL}{\ensuremath{\mathcal{L}}\xspace}

\newcommand{\CN}{\ensuremath{\mathcal{N}}\xspace}

\newcommand{\CS}{\ensuremath{\mathcal{S}}\xspace}
\newcommand{\CT}{\ensuremath{\mathcal{T}}\xspace}

\def \ifempty#1{\def\temp{#1} \ifx\temp\empty }
\newcommand{\Var}{\mathrm{Var}}

\newcommand{\bigO}{\smash{\ensuremath{O}}}

\newcommand{\p}{\ensuremath{\!+\!}}
\newcommand{\m}{\ensuremath{\!-\!}}
\newcommand{\e}{\ensuremath{\!=\!}}

\newtheorem{lemma}{Lemma}[section]
\newtheorem{corollary}[lemma]{Corollary}
\newtheorem{definition}[lemma]{Definition}

\newtheorem{remark}[lemma]{Remark}

\begin{document}
\title{An Analysis of Avalanche Consensus\thanks{Ignacio Amores-Sesar has been supported by the Swiss National Science Foundation (SNSF) under grant agreement Nr\@.~200021\_188443 (Advanced Consensus Protocols). Philipp Schneider has been supported by a grant from Avalanche, Inc.\ to the University of Bern.}}

\author[1,2]{Ignacio Amores-Sesar}
\author[1,3]{Christian Cachin}
\author[1,4]{Philipp Schneider}

\affil[1]{University of Bern} 
\affil[2]{\href{ignacio.amores@unibe.ch}{ignacio.amores@unibe.ch}} 
\affil[3]{\href{christian.cachin@unibe.ch}{christian.cachin@unibe.ch}} 
\affil[4]{\href{philipp.schneider2@unibe.ch}{philipp.schneider2@unibe.ch}}
 
\date{}
\maketitle              
\begin{abstract}
A family of leaderless, decentralized consensus protocols, called Snow consensus was introduced in a recent whitepaper by Yin et al.
These protocols address limitations of existing consensus methods, such as those using proof-of-work or quorums, by utilizing randomization and maintaining some level of resilience against Byzantine participants. Crucially, Snow consensus underpins the Avalanche blockchain, which provides a popular cryptocurrency and a platform for running smart contracts. 
	  
Snow consensus algorithms are built on a natural, randomized routine, whereby participants continuously sample subsets of others and adopt an observed majority value until consensus is achieved. Additionally, Snow consensus defines conditions based on participants' local views and security parameters. These conditions indicate when a party can confidently finalize its local value, knowing it will be adopted by honest participants.
	  
Although Snow consensus algorithms can be formulated concisely, there is a complex interaction between randomization, adversarial influence, and security parameters, which requires a formal analysis of their security and liveness. 
Snow protocols form the foundation for Avalanche-type blockchains, and this work aims to increase our understanding of such protocols by providing insights into their liveness and safety characteristics. 	  
First, we analyze these Snow protocols in terms of latency and security. Second, we expose a design issue where the trade-off between these two is unfavorable. Third, we propose a modification of the original protocol where this trade-off is much more favorable.
\end{abstract}

\section{Introduction}

Establishing consensus is one of the most fundamental tasks in distributed computing, for instance to implement atomic broadcast, to synchronize processes, or to elect leaders. Distributed blockchains and in particular cryptocurrencies rely on consensus to ensure proper operation, and therefore trust in these systems, which has put increased focus on new kinds of consensus algorithms. 
In the consensus problem we consider $n$ parties of which some are potentially faulty. Every party has some input value, which we often refer to as opinion in this article. We say that a protocol that coordinates communication and local computations of all parties {solves consensus} when everyone agrees on a single opinion, which was also the input of at least one party.

The consensus problem becomes challenging in the presence of Byzantine faults, i.e., parties that can deviate from the protocol. In particular, reaching consensus is impossible deterministically in the asynchronous setting with even one fault~\cite{Fischer1985} and in the synchronous setting with one third or more Byzantine faults~\cite{Fischer1986}, except if one would use cryptographic signatures. 
Real protocols that solve consensus in the context of blockchains have to navigate around these impossibility results, while also optimizing other criteria; chief among them are the latency until a transaction is finalized, the throughput of transactions, resource consumption, scalability, and resiliency against adversarial parties, which often necessitates trade-offs~\cite{DBLP:journals/sigact/GilbertL02,Buterin2017}.
 
One particularly simple consensus design sacrifices determinism and works along the following principle. Each party continually samples random subsets of other parties and adjusts its opinion based on the observed sample according to certain rules. 
There has been extensive research that explores such mechanisms, see the recent survey~\cite{Becchetti2020}. It has been shown for a subset of protocols of this type that they can be expected to converge very rapidly to a state of stable consensus (with a limited adversary, a bounded number of opinions, and in a synchronous network) \cite{Doerr2011,Becchetti2017,Ghaffari2018,Becchetti2016,Elsaesser2017}. Additionally, such mechanisms have the advantage that parties need only few such samplings to be relatively certain what the consensus opinion will be, resulting in near-linear message complexity.

A whitepaper \cite{Rocket2019} released in 2019 exploits this design and introduces the Avalanche protocol, which forms the basis of the Avalanche blockchain infrastructure and its services. Avalanche gained popularity and reach due to competitive  characteristics in the performance spectrum of latency, throughput, scalability, and resource consumption \cite{Gramoli2023}.\footnote{For instance, through its AVAX token, Avalanche ranks in the top 10 among the ``Layer-1'' blockchains by market capitalization (as of December 2023).}
In particular, \cite{Rocket2019} introduces the \emph{Snow family} of binary consensus protocols that build on this principle of random samplings, which can be adapted to maintain consistency of the corresponding Avalanche blockchain network.

The simplest protocol of this family, known as \emph{Slush}, works as follows. Each party continuously samples the opinion of $k\geq 2$ others; if such a sampling contains an opinion different from its own at least $\alpha$ times (for some $\alpha > k/2$) then the party adopts this opinion as its own.
Slush can be considered as a self-organizing mechanism that it is likely to converge to a stable consensus relatively quickly and remain there, even in the presence of a limited number of parties that deviate from the protocol.
The whitepaper also introduces the \emph{Snowflake} and \emph{Snowball protocols}, which add mechanisms to finalize an opinion of a node based on past queries which reflects how stable the observed majority is. The level of confidence in a majority can be controlled with a security parameter $\beta$.

The complex interaction between performance characteristics, security level, and the involved parameters $k$, $\alpha$, and $\beta$ makes the analysis of Snow-type consensus protocols challenging. The whitepaper \cite{Rocket2019} relies primarily on empirical observations and informal explanations to motivate its design choices.  
Currently, a formal understanding of the performance and security characteristics of Snow protocols is lacking.

\paragraph{Overview and Contributions.}

We focus on bridging the gap in understanding of Snow consensus protocols, which we consider as a necessary first step for an encompassing analysis of the complete Avalanche blockchain protocol, which builds upon Snow consensus (although the Avalanche protocol itself is beyond the scope of this work). 
First we explore the performance of Snow protocols, beginning with the self-organizing, binary consensus mechanism of Slush. 
In Section \ref{sec:slush_dynamics}, we express the progress toward a stable consensus per round depending on the distribution of opinions and parameters $k\geq 2$ and $\alpha > \tfrac{k}{2}$, which gives insights into the evolution of the system (cf.\ Figure \ref{fig:delta} for a visualization).

In Section \ref{sec:slush_bounds}, we show that coming close to a consensus already requires a minimum of \smash{$\Omega\big(\tfrac{\log n}{\log k}\big)$} rounds, even in the absence of adversarial influence (see Theorem \ref{thm:slush_lower_bound}, a simpler, weaker form is given in Corollary~\ref{cor:slush_lower_bound}). Furthermore, we generalize upper bounds from the so-called Median and 3-Majority consensus protocols \cite{Doerr2011,Becchetti2017} in the Gossip model (discussed in Section \ref{sec:slush_bounds}) and establish that Slush reaches a stable consensus in $\bigO(\log n)$ rounds (see Theorem \ref{thm:slush_upper_bound}), which holds even when an adversary can influence up to $\bigO\big(\!\sqrt{n}\big)$ parties.

We interpret these results in the following way. Even assuming that the performance of Slush matches the lower bound, increasing the parameter $k$ yields only a limited speed-up of $\bigO(\log k)$ (usually $k \ll n$). Furthermore, since the message complexity per round is $\Theta(kn)$ Slush has the advantage  of near-linear message complexity for small $k$, which is negated if $k$ increases significantly (e.g., sampling sizes close to $n$). We conclude that values higher than $k=20$ as suggested originally \cite{Rocket2019} have diminishing benefit and an unfavorable trade-off in terms of message complexity. In Section \ref{sec:snowball_dynamics} we show that the lower bound  \smash{$\Omega\big(\tfrac{\log n}{\log k}\big)$} rounds extends to Snowflake and Snowball.

In Section  \ref{sec:security_ball_flake} we analyze the security mechanisms of the Snow protocols, that deal with the possibility of failing to achieve consensus due to the randomized nature of the algorithm or adversarial influence.
The protocol provides a security parameter $\beta$ to control the probability of such a failure. This has an unfavorable trade-off as we show in Section \ref{sec:security_ball_flake}, specifically, a negligible probability of failure (w.r.t.\ $\beta$) and a latency to finalize a value that is at most polynomial in $\beta$ are mutually exclusive (see Theorem \ref{thm:impossibility}).
In Section \ref{sec:blizzard} we propose a solution for this issue by introducing an alternative protocol. It replaces the security mechanism of Snowball with a simple mechanism that achieves security with all but negligible probability (w.r.t.\ $\beta$) in $\bigO(\beta + \log n)$ rounds (see Theorem \ref{thm:security_blizzard} and Corollary \ref{cor:security_blizzard}).

\section{Preliminaries}
\label{sec:prelim}  

Before moving to the technical parts of the analysis, we introduce the definitions and modeling assumptions that we use throughout the paper. In part, this work aims to be a supplementary of the whitepaper~\cite{Rocket2019} in the style of other theoretical works that study randomized, self-stabilizing consensus protocols  \cite{Doerr2011,Becchetti2017,Ghaffari2018,Becchetti2016,Elsaesser2017}. Therefore our nomenclature, definitions and modeling assumptions are a composition of those.

\subsection{Model}
\label{sec:model}

\paragraph{Communication.}
\label{sec:adversary}
We consider a fully connected network of
$n$ parties with identifiers $\mathcal N = \{1, \dots, n\}$. Parties communicate by sending point to point messages. For the message transfer we assume the synchronous message passing model, where  there is a fixed period of time until any given message is delivered. In fact, the synchronous setting allows to assume that time is slotted into discrete rounds and all messages sent in the previous round have arrived by the next round. This also allows us to use the number of rounds as a proxy for algorithm running time.

\paragraph{Consensus.}

In the general problem setup there are $m$ opinions in the network and each party has an initial opinion, however note that in the context of Snow protocols we typically have $m=2$. Snow protocols can be seen as self-stabilizing protocols and we define a stable state where almost all parties have the same opinion and the likelihood to revert from this state is low (see Section \ref{sec:blizzard} for some properties of this stable state).

\begin{definition}[State of Stable Consensus]\label{def:consensus_state}
	The system is in a state of stable consensus if at least $n \m o(n)$ parties have the same opinion.
\end{definition}

Randomization or the presence of an adversary implies that at any point in time there is a non-zero chance that a stable consensus is reverted. Therefore, in the context of blockchain applications, parties need to eventually finalize or \textit{decide} on an opinion. In that sense we define the consensus problem as follows.

\begin{definition}[Consensus Problem]\label{def:consensus}
	A protocol solves this problem if the following conditions are satisfied.
	\begin{description}
		\item[Termination:] Every party eventually \emph{decides} on some opinion.
		\item[Validity:] If all parties propose the same value, then all parties decide on that value.
		\item[Integrity:] No party \emph{decides} twice.
		\item[Agreement:] No two parties \emph{decide} differently.
	\end{description}
\end{definition}

We consider this consensus problem under the influence  of the following adversary. 

\begin{definition}[$F$-Bounded Adversary]
	\label{def:adversary}
	An $F$-bounded adversary can set the opinion of up to $F$ (undecided) parties at the beginning of each round to one of the $m$ opinions.
\end{definition}

\paragraph{Randomization, Security and Latency.}
The consensus protocols we consider in this work are randomized, and we work with some standard definitions that we summarize in Appendix \ref{sec:probabilistic_concepts}. In particular we also consider an $F$-bounded adversary, which introduces a non-zero chance to delay a consensus for any fixed period of time. Hence, we can only hope to make the probability of failure of such a protocol negligibly small, while at the same time maintaining a reasonable latency (i.e., number of rounds until consensus). To connect the notions of failure probability and latency we make the following provisions.

Let $E$ be an event or condition during the execution of some protocol $\mathcal P$, e.g., $E$ describes the event that $\mathcal P$ successfully establishes consensus, see Definition \ref{def:consensus}.
The protocols we are investigating typically depend on so called security parameters. Increasing the security parameter increases the likelihood of success but typically has detrimental effects on the running time. 
To quantify this, we formally define further below what it means that $E$ holds with all but negligible probability with respect to protocol $\mathcal P$.  Roughly speaking, as we increase the running time of a protocol measured in the security parameter $\lambda$, the probability that some condition $E$ (e.g., consensus) has not been established yet, decreases super-polynomially in $\lambda$.

\begin{definition}[Negligible Probability, Security Parameter]
	\label{def:negligible}	
	An event $E$ holds with all but negligible probability or equivalently its complement $\overline E$ has negligible probability in a protocol $\mathcal P$, if the following holds. There exists a polynomial $\rho$ such that for any polynomial $\pi$ there exists a value $\lambda_0 \geq 0$, such that for all $\lambda \geq \lambda_0$ the following holds. If $\,\mathcal P$ is executed for at least $\rho(\lambda)$ rounds it holds that $\P(\overline E) \leq 1/\pi(\lambda)$, i.e., $\P(E) > 1-1/\pi(\lambda)$. 	
	We call the value $\lambda \geq \lambda_0$ a security parameter.
\end{definition}

In this article, we often talk about randomly sampling a set of $k$ parties, which means that we take a uniform random sample of size $k$ from the set of all $n$ parties \textit{with repetition}, which is modeled by the binomial distribution. Note that sampling \textit{without repetition}, represented by the hyper-geometric distribution, approaches the binomial distribution as  the ratio $n/k$ becomes larger. Therefore, our results approximate the case without repetition well for $n\gg k$ (the usual case). The assumption of uniform sampling with repetition makes our analysis much more feasible, in particular it removes undesirable marginal cases (e.g.\ $k$-samples for $n < k$) allowing use to use continuous functions to describe certain aspects of the system.

\subsection{The Snow Family}
\label{sec:protocols}

The Snow family consists of three consensus protocols based on random sampling instead of the traditional quorum intersection. This approach allows the snow family to reduce the message complexity by sending messages to a constant number of parties each round. 
Here, we provide an informal description of each protocol, whereas a more detailed pseudocode can be found in Appendix~\ref{apx:slush_code}.

\paragraph{Slush.}
\label{sec:slush}

 The first protocol in the Snow family is Slush (Algorithm~\ref{algo:slush}). An honest party $j$ runs the Slush protocol in local rounds, however in the general protocol every party may have a different round value. Party $j$ starts the round by randomly sampling $k$ parties for their opinion, with $k$ a constant value. 
 If at least $\alpha$ parties respond with the same opinion $b$, party $j$ adopts it as its own opinion and starts a new round. The value $\alpha$ must be strictly larger than $\frac{k}{2}$. If there is no $\alpha$-majority, party $j$ keeps its opinion and starts a new round. The Slush algorithm has a hard-coded number of rounds defining the protocol's end. When party $j$ reaches the maximum round, it $\op{decides}(b)$ its candidate value $b$. Note that in our analysis in Section \ref{sec:slush_dynamics} and \ref{sec:slush_bounds} we analyze the time until Slush reaches a state of stable consensus (Definition \ref{def:consensus_state}) and assume that this hard coded maximum round does not exist or is sufficiently large to not play a role.

\paragraph{Snowflake.}
\label{sec:snowflake}
The main limitation of the Slush algorithm is the hard-coded number of rounds. The number of rounds needs to be relatively high to guarantee consensus even in the worst case (which the case when the network starts in a bivalent state: half of the parties $\op{proposes}(0)$ and the other half $\op{proposes}(1)$). Snowflake aims to address this issue by modifying the termination condition of Slush.

In the Snowflake protocol (Algorithm~\ref{algo:snowflake}), party $j$ counts the number of consecutive queries with an $\alpha$-majority for opinion $b$. If $j$ observes $\beta$ consecutive rounds with $\alpha$-majority for $b$, party $j$ $\op{decides}(b)$. The intuition behind this termination rule is the following: the probability of obtaining $\beta$ consecutive $\alpha$-majorities for opinion $b$ is small when expressed as a function of $\beta$, unless almost every party has $b$ as candidate value in the network. This termination rule allows for an adaptive running time based on the state of the network. Looking ahead, we show how this termination rule forces the Snowflake protocol to choose between a high confidence in the agreement property and polynomial running time (Theorem~\ref{thm:impossibility}). 
 
\paragraph{Snowball.}
\label{sec:snowball}

The Snowball protocol (Algorithm~\ref{algo:snowball}) introduces another modification how parties change their opinion. 
In Snowflake, the change of opinion is only based on the outcome of the last query, i.e., it is stateless. By contrast, in Snowball, party $j$ considers the past queries in order to decide whether to change its opinion or not. Party $j$ changes its opinion value from $b$ to $b'$ when the number of $\alpha$-majorities for $b'$ surpasses the number of $\alpha$-majorities for value $b$ since the beginning of the execution. The idea behind considering the whole history of the protocol is to make it less likely for a party with opinion $b$ to switch to $b'$ when the prevalent opinion in the network is $b$, thus possibly reducing the number of rounds until termination. Looking ahead, we show that this routine does not reduce the number of rounds until termination in expectation (Lemma~\ref{lemma:comparison}).

\paragraph{Avalanche.}
The Snow consensus protocols serve as foundation for the Avalanche consensus~\cite{Rocket2019,AmoresSesar2022}. Avalanche employs a classification system to group transactions into conflicting sets and subsequently applies a tailored adaptation of the Snowball algorithm to each of these conflict sets. To optimize communication efficiency, Avalanche establishes connections between distinct instances of the Snowball consensus, enabling the reuse of messages and, consequently, reducing message complexity. However, due to these interdependencies, Avalanche is unable to inherit the liveness properties from the Snow family~\cite{AmoresSesar2022}. Nevertheless, it is noteworthy that the security of Avalanche remains equivalent to that of the Snowball protocol~\cite{Rocket2019}.

\subsection{Related Work}
\label{sec:related_work}

The Avalanche protocol was introduced fairly recently thus research into this protocol is limited. The whitepaper \cite{Rocket2019} gives an iterative presentation of its algorithms and concepts, in particular the Snow protocols for binary consensus on which it then builds its Avalanche protocol in the UTXO model. This is supplemented with explanations about the design decisions and empiric data that highlights the protocols' performance in terms of latency, throughput and Byzantine resilience. Subsequently, the article \cite{AmoresSesar2022} provided a formal description of the Avalanche protocol. They also showcase a vulnerability (that has since been addressed by subsequent versions of Avalanche) that is specific to the Avalanche protocol, where a single malicious party can delay acceptance of a transaction and proposes a modification that prohibits this attack. Part of the reason that an encompassing analysis of Avalanche is outside the scope for this work, is that it is currently still evolving, for instance recently transitioning from a directed acyclic graph (DAG) to a chain, providing a total order for transactions as opposed to a partial order.

Self-stabilizing consensus protocols based on random samplings have been investigated much earlier in message passing models, motivated by  the so called GOSSIP model.\footnote{In the GOSSIP model, nodes can contact a few random neighbors in a graph.} A particular strain of such protocols that attained some focus in the past are the so called 3-Majority, the 2-choices and the Median protocols \cite{Becchetti2020}. In the 3-Majority protocol parties sample 3 random others and adopt the majority opinion using the first sampled parties' opinion in case of a tie. The 2-Choices protocol works similar, but only 2 parties are sampled with the third being the party itself which also provides the default opinion. In the Median protocol a party samples 2 others and adopts the median value among theirs and their own (which requires a total order on the opinions).  
There has been a plethora of work on the analysis of these and similar light weight protocols based on random sampling, a selection of those are \cite{Doerr2011,Cruise2014,Cooper2014,Becchetti2017,Becchetti2016,Ghaffari2018}, see also the survey \cite{Becchetti2020} for an overview. These works usually focus on analyzing the time to consensus with respect to the initial number of opinions in the network, sometimes also on the required initial bias of the network in case a consensus on the initial majority is desired (plurality consensus).

In contrast to this work, these articles focus on samplings of size at most 3, analysis of the dynamics for the whole spectrum of $k,\alpha$ have, to the best of our knowledge, so far not been attempted.\footnote{Metrics other than the round complexity have been considered for the binary $k$-Majority protocol which relates to Slush \cite{Cruciani2021}.}
Crucially, in case the number of opinions is constant, all these works arrive at $\bigO(\log n)$ rounds until a state of stable consensus is attained with high probability. Interestingly, in the binary case the 3-Majority, the 2-choices and the Median protocols can be all be related to special cases of Slush. Besides analyzing security aspects of snow protocols, one of the main contributions of this work is to show how the dynamics of such sampling based protocols behave in the size  $k$  of those samplings.

\section{Dynamics of Slush}
\label{sec:slush_dynamics}

The whitepaper \cite{Rocket2019} observes that the Slush consensus protocol converges to a stable consensus very fast in practice. Concrete claims are made pertaining to the time to consensus, but no conclusive proof is given.
In this section we analyze the the rate of convergence of Slush towards a consensus, which will later also inform the rate of convergence of Snowflake and Snowflake (Section \ref{sec:snowball_dynamics}). 

\subsection{Expected Rate of Progress of Slush}
\label{sec:sprogress}
We start by investigating the expected rate of progress of Slush, which characterizes the dynamics of Slush and how it depends on the parameters $\alpha$ and $k$. We will later show that other Snow protocols (Snowflake, Snowball) behave similar in terms of the required number of rounds to consensus. 

We make the following definitions and assumptions. First we assume that all parties have an initial opinion 0 or 1, so no party has initially the opinion $\bot$ (the case where there exist parties with opinion $\bot$ can be disregarded for the lower bound and is handled separately for the upper bound). Recall that the set of parties $\mathcal N$ is numbered from $1$ to $n$ and assume rounds are numbered $0,1,2,\dots$.

\begin{itemize}	
	\item Let $X_{ij} \in \{0,1\}$ be the current opinion of party $j$ after round $i$ of the Slush protocol was executed ($X_{0,j}$ describes the initial opinion of party $j$). 
	
	\item Let $Y_{ij}$ be the number of replies with opinion 1 that party $j$ obtains in its sample of $k$ parties in round $i$. Note that $Y_{ij} \sim \text{Bin}(k,p_i)$.
	
	\item Let the state of the network be $S_i := \sum_{j=1}^{n} X_{ij}$, which describes the total number of parties whose current opinion is 1. 
	
	\item Let $p_i := S_i/n$ be the relative share of parties with opinion 1 in round $i$, which corresponds the probability that a sampled party has state 1.
	
	\item For $i \geq 1$, we define as $\Delta_i := S_i-S_{i-1}$, i.e., the absolute progress to 1-consensus (or 0-consensus for negative values). 
	
	\item Let $\delta_i := \E(\Delta_i)/n$ be the expected relative progress in round $i$. We will later show that $\delta_i$ can also be expressed as a function $\delta : [0,1] \to \mathbb R$ that only depends on $p_i$ (when viewing $k,\alpha$ as fixed values), such that $\delta_i = \delta(p_i)$. Subsequently, we establish a relation between the $\delta(p_i)$ for varying parameters $k,\alpha$, in which case we denote it as $\delta^{k,\alpha}(p_i)$ (however, for conciseness we will refrain using this superscript whenever possible, in particular when only single values for $k$ and $\alpha$ are involved).
\end{itemize}

Note that for $i \geq 1$, the quantities $X_{ij}$, $Y_{ij}$, $S_i, \Delta_i, $ are random variables. This is not the case for the \textit{expected} relative progress $\delta_i$, which, for fixed $\alpha,k$, can be expressed only in terms of $p_i$, i.e., $\delta_i$ can be expressed as a function $\delta(p_i)$ that depends only on $p_i$.

\begin{restatable}{lemma}{lemdelta}
\label{lem:delta}
	Let $k/2 < \alpha \leq k$. Then 	
	\[
		\delta_i = \delta(p_i) := \sum_{\ell = \alpha}^{k} \binom{k}{\ell} \Big[p_i^\ell (1 \m p_i)^{k-\ell+1} - (1 \m p_i)^\ell p_i^{k-\ell+1}\Big].
	\]
\end{restatable}

\begin{proof}	
	In a given round $i$ a party $j$ with can observe at least $\alpha$ times the opposite opinion in its query, in which case it switches, or not (observing at least $\alpha$ of both opinions is precluded due to $\alpha > k/2$). Note that for $\delta_i$ only events where parties switch their opinion are relevant. For some party $j$ we have
	\begin{align*}
		\P\big(X_{ij} = 1 \mid X_{i-1,j} = 0\big) & = \P\big(Y_{ij} \geq \alpha \big) \\ 
		& = \sum_{\ell = \alpha}^{k} \binom{k}{\ell} p_i^\ell (1 \m p_i)^{k-\ell},\\
		\P\big(X_{ij} = 0 \mid X_{i-1,j} = 1\big) & = \P\big(Y_{ij} \leq k-\alpha \big) \\ 
		& = \sum_{\ell = 0}^{k-\alpha} \binom{k}{\ell} p_i^\ell (1 \m p_i)^{k-\ell} = \sum_{\ell = \alpha}^{k} \binom{k}{\ell} (1 \m p_i)^\ell p_i^{k-\ell},
	\end{align*}
	where the last step is due to symmetry (see also Appendix \ref{sec:binom}).
	Let $\delta^0_i$ be the probability that a randomly selected party switches to from 1 to 0 and let $\delta_i^1$ be analogous probability for switching from 0 to 1. These can also be interpreted as expected portions of parties switching from 1 to 0 and vice versa. We obtain	
	\begin{align*}
		\delta^0_i& = \P\big( X_{ij} \e 0 \cap  X_{i-1,j} \e 1\big) = (1 \m p_i) \P\big(X_{ij} \e 0\mid X_{i-1,j} \e 1\big) = \sum_{\ell = \alpha}^{k} \binom{k}{\ell} p_i^\ell (1 \m p_i)^{k-\ell+1},\\
		\delta^1_i & = \P\big( X_{ij} \e 1 \cap  X_{i-1,j} \e 0\big) = p_i \cdot \P\big(X_{ij} \e 1 \mid X_{i-1,j} \e 0\big) = \sum_{\ell = \alpha}^{k} \binom{k}{\ell} (1 \m p_i)^\ell p_i^{k-\ell+1}.		
	\end{align*}
	The random variable $\Delta_i$  is determined by the number of parties that switch from 0 to 1 minus those that switch from 1 to 0. Since the events $ X_{ij} \e 0 \cap  X_{i-1,j} \e 1$ and  $X_{ij} \e 1 \cap  X_{i-1,j} \e 0$ are disjoint, we have that $E(\Delta_i) = n\delta^1_i-n\delta^0_i$. Thus $\delta_i = n E(\Delta_i) =  \delta^1_i - \delta^0_i$ and the claim follows.	
\end{proof}

\begin{figure}
	\centering	
	\begin{subfigure}{0.49\textwidth}
		\centering
		\includegraphics[width=\linewidth]{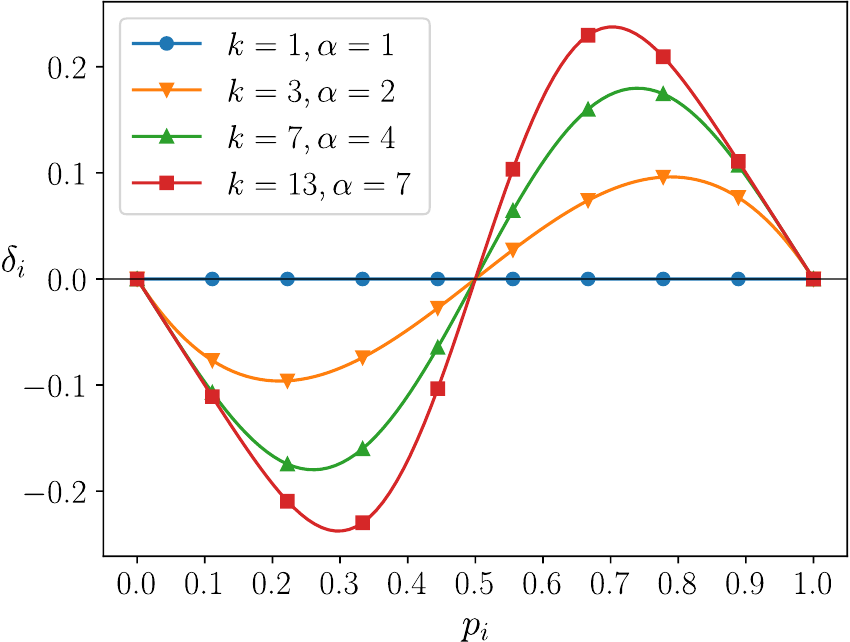}
		\caption{$k=2\alpha \m 1$ with various values $\alpha$}
		\label{fig:delta_a}
	\end{subfigure}
	\hfill
	\begin{subfigure}{0.49\textwidth}
		\centering
		\includegraphics[width=\linewidth]{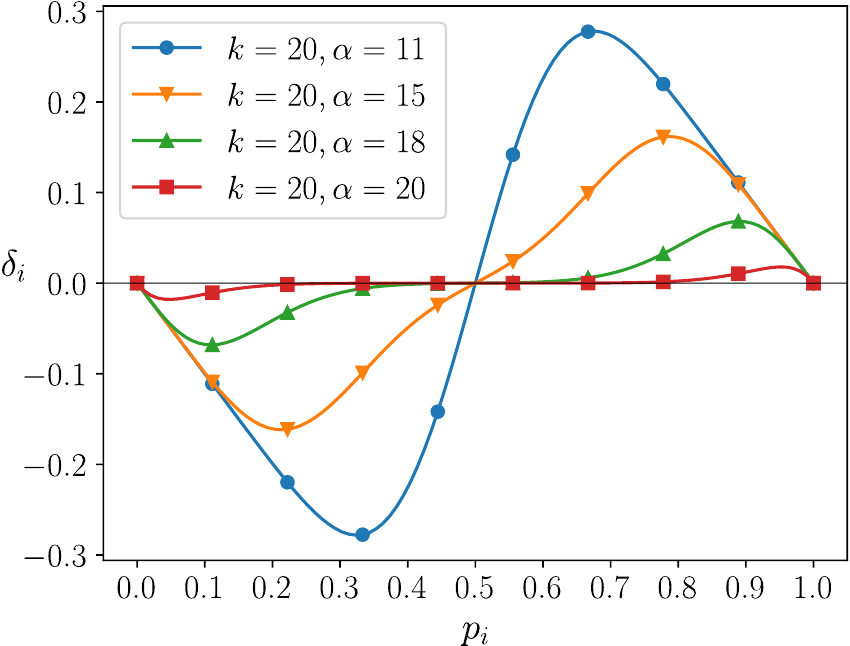}
		\caption{$k=20$ with various values $\alpha$}
		\label{fig:delta_b}
	\end{subfigure}
	
	\caption{Plots of $\delta(p_i)$ for different parameters $k$ and $\alpha$. For $k=2\alpha\m1$ the expected  progress for larger $\alpha$ dominates those for smaller (note that in the extreme case $k=\alpha=1$ there is no expected progress). For fixed $k$ the opposite is true. The combination $k=20, \alpha =15$ was suggested by the whitepaper \cite{Rocket2019}. Note that $\delta(p_i)$ is point-symmetric with respect to the point $(\tfrac{1}{2},0)$.}
	\label{fig:delta}
\end{figure}

\subsection{Mapping out the Dependency of \boldmath$\delta_i$ on \boldmath$k$ and \boldmath$\alpha$}

Recall that we define \smash{$\delta^{k,\alpha}(p_i)$} as the function defining the rate of progress for  parameters $\alpha,k$ depending on $p_i$.
It is interesting that for $k=2\alpha-1$ we have \smash{$\delta^{k,\alpha} = \delta^{k,\alpha+1}$} as the first summand in the expression \smash{$\delta^{k,\alpha}(p_i)$} from Lemma \ref{lem:delta} is zero. Another interesting observation is that in the marginal case $k=\alpha=1$, we have $\delta^{k,\alpha} = 0$ (see Figure \ref{fig:delta_a}), i.e., there is no expected progress and Slush essentially degenerates into a random walk, and it is not too hard to show that in this case Slush takes $\Omega(n^2)$ rounds in expectation.

In this section we establish non-trivial relations and claims for \smash{$\delta^{k,\alpha}(p_i)$}. First, we show that \smash{$\delta^{k,\alpha}(p_i)$} is always larger (in absolute value) than \smash{$\delta^{k,\alpha'}(p_i)$} for $\alpha' > \alpha$ (see Figure \ref{fig:delta_b} for a visualized example of this claim).  This means that choosing the smallest $\alpha$ with $\alpha > k/2$ (i.e., \smash{$\alpha = \lceil \frac{k+1}{2}\rceil$}) is best in terms of the expected rate of progress $\delta^{k,\alpha}(p_i)$. While this is useful on its own regarding the choice of $\alpha$ in practice, we utilize this later to essentially eliminate $\alpha$ from the of analysis for the lower bound of the rate of convergence to consensus.

\begin{restatable}{lemma}{lemdeltadomination}
	\label{lem:delta_domination}
	For fixed $k$, let $k/2 < \alpha \leq k$ and consider $\alpha' > \alpha$. Then for any $p_i$ we have $|\delta^{k,\alpha}(p_i)| \geq |\delta^{k,\alpha'}(p_i)|$.	
\end{restatable}

\begin{proof}
	The lemma is a corollary of Lemma  \ref{lem:delta}, where we notice that all summands of \smash{$\delta(p_i) = \delta^{k,\alpha}(p_i)$} are either positive for $p_i > 1/2$ or negative for  $p_i < 1/2$ or 0 for $p_i = 1/2$. Since all summands of \smash{$\delta^{k,\alpha'}(p_i)$} occur in \smash{$\delta^{k,\alpha}(p_i)$} and all have the same sign for some fixed $p_i$, we have the claim.
\end{proof}

In particular, for $\alpha = \lceil \frac{k+1}{2}\rceil$ we can express $\delta^{k,\alpha}(p_i)$ in a different form, which we will use frequently in subsequent lemmas to simplify (or enable) subsequent proofs.

\begin{restatable}{lemma}{lemdeltashortform}
	\label{lem:delta_short_form}
	Let $\alpha = \lceil \frac{k+1}{2}\rceil$. Then
	\[
		\begin{cases}
			\delta^{k,\alpha}(p_i) = \big[\sum_{\ell = \alpha}^{k} \binom{k}{\ell} p_i^\ell (1 \m p_i)^{k-\ell}\big] - p_i, & \quad k \text{ odd}\\
			\delta^{k,\alpha}(p_i) = \big[\sum_{\ell = \alpha}^{k} \binom{k}{\ell} p_i^\ell (1 \m p_i)^{k-\ell}\big] - p_i + \binom{k}{\alpha-1} p_i^\alpha (1 \m p_i)^{\alpha-1}, & \quad k \text{ even.}
		\end{cases}
	\]
\end{restatable}

\begin{proof}
	We start out with the expression from Lemma \ref{lem:delta}.
	\begin{align*}
		\delta^{k,\alpha}(p_i) = & \sum_{\ell = \alpha}^{k} \binom{k}{\ell} \Big[p_i^\ell (1 \m p_i)^{k-\ell+1} - (1 \m p_i)^\ell p_i^{k-\ell+1}\Big]\\
		= & \sum_{\ell = \alpha}^{k} \binom{k}{\ell} p_i^\ell (1 \m p_i)^{k-\ell+1} - \sum_{\ell = \alpha}^{k} \binom{k}{\ell} (1 \m p_i)^\ell p_i^{k-\ell+1}\\
		= & \sum_{\ell = \alpha}^{k} \binom{k}{\ell} \Big[p_i^\ell (1 \m p_i)^{k-\ell} - p_i^{\ell +1} (1 \m p_i)^{k-\ell} \Big] - \sum_{\ell = \alpha}^{k} \binom{k}{\ell} (1 \m p_i)^\ell p_i^{k-\ell+1}\\
		= & \sum_{\ell = \alpha}^{k} \binom{k}{\ell} p_i^\ell (1 \m p_i)^{k-\ell} - \sum_{\ell = \alpha}^{k} \binom{k}{\ell} p_i^{\ell +1} (1 \m p_i)^{k-\ell} - \sum_{\ell = \alpha}^{k} \binom{k}{\ell} (1 \m p_i)^\ell p_i^{k-\ell+1}
	\end{align*}	
	The proof forks into cases by the parity of $k$. Consider the case that $k$ is odd. Then we have $k = 2 \alpha-1$. Continuing from above, we substitute $\ell' := k - \ell$  in the middle sum, i.e., $\ell'$ goes from 0 to $k-\alpha = \alpha-1$.
	\begin{align*}
		= & \sum_{\ell = \alpha}^{k} \binom{k}{\ell} p_i^\ell (1 \m p_i)^{k-\ell} - \sum_{\ell'  = 0}^{\alpha-1} \binom{k}{k-\ell'} p_i^{k-\ell' +1} (1 \m p_i)^{\ell'}  - \sum_{\ell = \alpha}^{k} \binom{k}{\ell} (1 \m p_i)^\ell p_i^{k-\ell+1} \\
		= & \sum_{\ell = \alpha}^{k} \binom{k}{\ell} p_i^\ell (1 \m p_i)^{k-\ell} - \sum_{\ell' = 0}^{\alpha-1} \binom{k}{\ell'} p_i^{k-\ell' +1} (1 \m p_i)^{\ell'}  - \sum_{\ell = \alpha}^{k} \binom{k}{\ell} (1 \m p_i)^\ell p_i^{k-\ell+1} \\
		= & \sum_{\ell = \alpha}^{k} \binom{k}{\ell} p_i^\ell (1 \m p_i)^{k-\ell} - p_i\underbrace{\sum_{\ell = 0}^{k} \binom{k}{\ell} (1 \m p_i)^\ell p_i^{k-\ell}}_{=1} = \sum_{\ell = \alpha}^{k} \binom{k}{\ell} p_i^\ell (1 \m p_i)^{k-\ell} - p_i.
	\end{align*}
	In the case that $k$ is even we have $k=2\alpha-2$. Again we substitute  $\ell' := k - \ell$  in the middle sum, that is, $\ell'$ goes from 0 to $k-\alpha = \alpha-2$.
	\begin{align*}
		= & \sum_{\ell = \alpha}^{k} \binom{k}{\ell} p_i^\ell (1 \m p_i)^{k-\ell} - \sum_{\ell'  = 0}^{\alpha-2} \binom{k}{k-\ell'} p_i^{k-\ell' +1} (1 \m p_i)^{\ell'}  - \sum_{\ell = \alpha}^{k} \binom{k}{\ell} (1 \m p_i)^\ell p_i^{k-\ell+1} \\
		= & \sum_{\ell = \alpha}^{k} \binom{k}{\ell} p_i^\ell (1 \m p_i)^{k-\ell} - \sum_{\ell' = 0}^{\alpha-2} \binom{k}{\ell'} p_i^{k-\ell' +1} (1 \m p_i)^{\ell'}  - \sum_{\ell = \alpha}^{k} \binom{k}{\ell} (1 \m p_i)^\ell p_i^{k-\ell+1} \\
		= & \sum_{\ell = \alpha}^{k} \binom{k}{\ell} p_i^\ell (1 \m p_i)^{k-\ell} - p_i\sum_{\ell = 0}^{k} \binom{k}{\ell} (1 \m p_i)^\ell p_i^{k-\ell} + \binom{k}{\alpha \m 1} p_i^{\alpha} (1 \m p_i)^{\alpha-1} \\
		= & \sum_{\ell = \alpha}^{k} \binom{k}{\ell} p_i^\ell (1 \m p_i)^{k-\ell} - p_i + \binom{k}{\alpha \m 1} p_i^{\alpha} (1 \m p_i)^{\alpha-1} .\tag*{\qedhere}
	\end{align*}
\end{proof}

The next two lemmas show, perhaps surprisingly, that for odd $k = 2\alpha \m  1$ the expected progress functions $\delta^{k,\alpha}(p_i)$ and $\delta^{k-1,\alpha}(p_i)$ coincide. When combined with Lemma~\ref{lem:delta_domination}, this will later, when we cover lower bounds, allow us to focus our analysis solely on odd values of the form $k = 2\alpha -1$ as the corresponding claim for even $k = 2 \alpha-2$ is implied.

\begin{restatable}{lemma}{lemrephrasesum}
	\label{lem:rephrase_sum}
	For $k  > \alpha \geq 1$ the following equation holds
	\[
		\sum_{\ell = \alpha}^{k} \binom{k}{\ell} p_i^\ell (1 \m p_i)^{k-\ell} = \Big[\sum_{\ell = \alpha}^{k \m 1} \binom{k \m 1}{\ell} p_i^\ell (1 \m p_i)^{k-\ell-1}\Big] + \binom{k \m 1}{\alpha \m 1} p_i^{\alpha} (1 \m p_i)^{k - \alpha}.
	\]
\end{restatable}

\begin{proof}
	\begin{align*}
		& \hspace*{-1cm}\sum_{\ell = \alpha}^{k} \binom{k}{\ell} p_i^\ell (1 \m p_i)^{k-\ell}\\
		= & \,\Big[\sum_{\ell = \alpha}^{k \m 1} \binom{k}{\ell} p_i^\ell (1 \m p_i)^{k-\ell}\Big] + p_i^k\\
		= & \,\Big[\sum_{\ell = \alpha}^{k \m 1} \Big\{\binom{k \m 1}{\ell} + \binom{k \m 1}{\ell \m 1}\Big\} p_i^\ell (1 \m p_i)^{k-\ell}\Big] + p_i^k \\
		= & \,\Big[\sum_{\ell = \alpha}^{k \m 1} \Big\{\binom{k \m 1}{\ell} + \binom{k \m 1}{\ell \m 1}\Big\} p_i^\ell (1 \m p_i)^{k-\ell}\Big] + p_i^k \\
		= & \,\Big[\sum_{\ell = \alpha}^{k \m 1} \binom{k \m 1}{\ell} p_i^\ell (1 \m p_i)^{k-\ell}\Big] + \Big[\sum_{\ell = \alpha}^{k \m 1} \binom{k \m 1}{\ell \m 1} p_i^\ell (1 \m p_i)^{k-\ell}\Big]  + p_i^k \\
		= & \,\Big[\sum_{\ell = \alpha}^{k \m 1} \binom{k \m 1}{\ell} p_i^\ell (1 \m p_i)^{k-\ell}\Big] + \Big[\sum_{\ell = \alpha -1}^{k \m 2} \binom{k \m 1}{\ell} p_i^{\ell+1} (1 \m p_i)^{k-\ell-1}\Big]  + p_i^k\\
		= & \,\Big[\sum_{\ell = \alpha}^{k \m 1} \binom{k \m 1}{\ell} p_i^\ell (1 \m p_i)^{k-\ell-1}\Big] - \Big[\sum_{\ell = \alpha}^{k \m 1} \binom{k \m 1}{\ell} p_i^{\ell+1} (1 \m p_i)^{k-\ell-1}\Big]\\ 
		& + \Big[\!\sum_{\ell = \alpha -1}^{k \m 2}\! \binom{k \m 1}{\ell} p_i^{\ell+1} (1 \m p_i)^{k-\ell-1}\Big]  + p_i^k \\
		= & \,\Big[\sum_{\ell = \alpha}^{k \m 1} \binom{k \m 1}{\ell} p_i^\ell (1 \m p_i)^{k-\ell-1}\Big] + \binom{k \m 1}{\alpha \m 1} p_i^{\alpha} (1 \m p_i)^{k-\alpha} - p_i^k + p_i^k \tag*{\qedhere}
	\end{align*}
\end{proof}

\begin{restatable}{lemma}{lemdeltaevenodd}
	\label{lem:delta_even_odd}
	Let $k = 2\alpha \m 1$ and $\alpha \geq 2$. Then $\delta^{k,\alpha} = \delta^{k-1,\alpha}$.
\end{restatable}

\begin{proof}
	The proof is implied by combining Lemma \ref{lem:delta_short_form} and Lemma \ref{lem:rephrase_sum}. We start with the expression of $\delta^{k,\alpha}(p_i)$ derived in Lemma \ref{lem:delta_short_form}. 
	\begin{align*}
		\delta^{k,\alpha}(p_i) \stackrel{\text{Lem.\ref{lem:delta_short_form}}}{=} \hspace*{-2.5mm}& \hspace*{3.5mm} \,\Big[\sum_{\ell = \alpha}^{k} \binom{k}{\ell} p_i^\ell (1 \m p_i)^{k-\ell}\Big] - p_i\\
		\stackrel{\text{Lem.\ref{lem:rephrase_sum}}}{=} \hspace*{-2.5mm}& \hspace*{3.5mm} \,\Big[\sum_{\ell = \alpha}^{k \m 1} \binom{k \m 1}{\ell} p_i^\ell (1 \m p_i)^{k-\ell-1}\Big] + \binom{k \m 1}{\alpha \m 1} p_i^{\alpha} (1 \m p_i)^{\alpha-1} - p_i  \\  
		\stackrel{\text{Lem.\ref{lem:delta_short_form}}}{=} \hspace*{-2.5mm}& \hspace*{3.5mm}\delta^{k-1,\alpha}(p_i) \tag*{\qedhere}
	\end{align*}
\end{proof}

As our final structural claim in this subsection, we show that for $k=2\alpha\m 1$ or $k=2\alpha$, the expected progress towards consensus \smash{$|\delta^{k,\alpha}(p_i)|$} does not decrease as $\alpha$ (and thereby $k$) increases, an example is given in Figure \ref{fig:delta_a}. We will use the this lemma later to extend an upper bound for the number of rounds to a consensus from small values of $\alpha$ to large (where $k=2\alpha$ or $k=2\alpha\m 1$ ).

\begin{restatable}{lemma}{lemdeltadominationtwo}
	\label{lem:delta_domination2}
	Let $k=2\alpha\m1$ and $k'=2\alpha'\m1$ for $\alpha > \alpha' > 1$. Then for any $p_i$ we have $|\delta^{k,\alpha}(p_i)| \geq |\delta^{k',\alpha'}(p_i)|$. The same holds for even $k=2(\alpha\m1)$ and $k'=2(\alpha'\m1)'$.
\end{restatable}

\begin{proof}
	We show the bound for $p_i \geq \tfrac{1}{2}$, the claim for $p_i \leq \tfrac{1}{2}$ holds by symmetry. Let $k = 2\alpha\m1$. We show that $\delta^{k,\alpha}(p_i) = \delta^{k-1,\alpha}(p_i) \geq \delta^{k-2,\alpha-1}(p_i) \geq \delta^{k-3,\alpha-1}(p_i)$, which implies the lemma for $\alpha' = \alpha-1$ and the rest follows by induction.
	\begin{align*}
		\delta^{k,\alpha}(p_i) \stackrel{\text{Lem.\ref{lem:delta_even_odd}}}{=} \hspace*{-2.5mm}& \hspace*{3.5mm} \delta^{k-1,\alpha}(p_i)\\
		\stackrel{\text{Lem.\ref{lem:rephrase_sum}}}{=} \hspace*{-2.5mm}& \hspace*{3.4mm} \Big[\sum_{\ell = \alpha}^{k \m 2} \binom{k \m 2}{\ell} p_i^\ell (1 \m p_i)^{k-\ell-2}\Big] + \binom{k \m 2}{\alpha \m 1} p_i^{\alpha} (1 \m p_i)^{k-\alpha} - p_i  \\  
		\stackrel{p_i\geq1/2}{\geq} \hspace*{-3.3mm}& \hspace*{3.3mm} \Big[\sum_{\ell = \alpha}^{k \m 2} \binom{k \m 2}{\ell} p_i^\ell (1 \m p_i)^{k-\ell-2}\Big] + \binom{k \m 2}{\alpha \m 1} p_i^{\alpha-1} (1 \m p_i)^{k-\alpha+1} - p_i  \\  
		= & \,\,\,\,\Big[\sum_{\ell = \alpha-1}^{k \m 2} \binom{k \m 2}{\ell} p_i^\ell (1 \m p_i)^{k-\ell-2}\Big]  - p_i  \\  
		\stackrel{\text{Lem.\ref{lem:delta_short_form}}}{=} \hspace*{-2.5mm}& \hspace*{3.5mm} \delta^{k-2,\alpha-1}(p_i)
		\stackrel{\text{Lem.\ref{lem:delta_even_odd}}}{=}  \delta^{k-3,\alpha-1}(p_i).\tag*{\qedhere}
	\end{align*}
\end{proof}

Note that while increasing $\alpha$ with $k=2\alpha\m 1$ or $k=2\alpha$ does in fact strictly increase $|\delta^{\alpha,k}(p_i)|$ thereby speeding up the expected time to consensus, this effect is rather limited, as we shall see in the next section.

\section{Bounding the Time to Consensus for Slush}
\label{sec:slush_bounds}

We will now show how the expected progress $\delta(p_i)$ can be used to obtain bounds for the number of rounds required to obtain consensus.

\subsection{Lower Bound}

As the system converges to consensus, arguably the most critical phase is when the network is roughly in balanced state, i.e., where fractions of parties with opinion 0 and 1 are roughly equal ($p_i \approx 1/2$) and where progress $\delta(p_i)$ towards consensus is close to 0, see Figure \ref{fig:delta}.

To lead us out of a potential perfect balance, the system can only rely pure randomness to gain some small initial imbalance, as the expected progress is 0. (The best one can hope for is a deviation of $p_i \approx 1/2 + {c}/{\sqrt{n}}$ for some constant $c$ within reasonable time bounds, due to the central limit theorem). After that initial perturbation the convergence to consensus crucially depends on how fast the expected progress for the next round grows in parameter $p_i$.  

Indicative for the change in progress  is the derivative of $\delta(p_i)$, whose upper bound is useful to analyze the case where the system moves to a 1-consensus (w.l.o.g., due to a symmetry argument).  Intuitively, this limits how fast the expected progress increases from an almost balanced state. We will first restrict ourselves to the case $k=2\alpha\m1$, as the previous section gives us all tools to extend this result to general $k$ and $\alpha$, as will be shown formally afterwards.

\begin{restatable}{lemma}{lemboundderivative}
	\label{lem:bound_derivative}
	Let $k=2\alpha-1$. For $p_i \geq 1/2$ it holds that $\frac{\partial\delta(p_i)}{\partial p_i}  \leq k-1$.
\end{restatable}

\begin{proof}
	We use the expression from Lemma \ref{lem:delta_short_form} and obtain the following derivative.	
	\begin{align*}
		\frac{\partial\delta(p_i)}{\partial p_i} = & \, \frac{\partial}{\partial p_i} \Big(\Big[\sum_{\ell = \alpha}^{k} \binom{k}{\ell} p_i^\ell (1 \m p_i)^{k-\ell}\Big] - p_i \Big) \\
		= & \, \Big[\sum_{\ell = \alpha}^{k} \binom{k}{\ell} \Big(\ell \cdot p_i^{\ell - 1} (1 \m p_i)^{k-\ell} - (k \m\ell) \cdot p_i^{\ell} (1 \m p_i)^{k-\ell-1}\Big)\Big] - 1
	\end{align*}	
	We evaluate at $p_i = 1/2$, then 
	\begin{align*}
		\frac{\partial\delta(\tfrac{1}{2})}{\partial p_i} = & \,   \frac{1}{2^{k-1}}\Big[\sum_{\ell = \alpha}^{k} \binom{k}{\ell} (2\ell-k)  \Big] - 1
		\leq  \,\frac{1}{2^{k-1}}\Big[\sum_{\ell = \alpha}^{k} \binom{k}{\ell} (2\ell-k) \Big] -1\\
		= & \,  \frac{1}{2^{k-1}}\Big[2\sum_{\ell = \alpha}^{k} \binom{k}{\ell}\ell - \frac{k}{2}\sum_{\ell = 0}^{k} \binom{k}{\ell} \Big] -1 \tag*{\small \textit{due to $k=2\alpha \m 1$ and $\binom{k}{\ell} = \binom{k}{k-\ell}$}}\\
		\leq & \,  \frac{1}{2^{k-1}}\Big[2\sum_{\ell = 0}^{k} \binom{k}{\ell}\ell - \frac{k}{2}\sum_{\ell = 0}^{k} \binom{k}{\ell} \Big]  -1\\
		= &\, \frac{1}{2^{k-1}}\Big[2k \cdot 2^{k-1} - \frac{k}{2} \cdot 2^{k} \Big]  -1= k \m1.
	\end{align*}	
	In the following we will also show that the second derivative \smash{$\frac{\partial\delta(p_i)}{\partial^2 p_i}$} is at most 0 for $p_i \geq \frac{1}{2}$. Combined with the above, this implies that the first derivative is \smash{$\frac{\partial\delta(p_i)}{\partial p_i} \leq k \m 1$} for any $p_i \geq \frac{1}{2}$. It remains to prove the claim about \smash{$\frac{\partial\delta(p_i)}{\partial^2 p_i}$}. This is a bit tedious and involves modifying binomial coefficients on the level of the definition $\binom{k}{\ell} = \frac{n!}{k!(n-k)!}$ (third step) and then shifting sum indices (fourth step).
	\begin{align*}
		\frac{\partial\delta(p_i)}{\partial^2 p_i} = & \, \frac{\partial}{\partial p_i} \Big[\sum_{\ell = \alpha}^{k} \binom{k}{\ell} \Big(\ell \cdot p_i^{\ell - 1} (1 \m p_i)^{k-\ell} - (k \m\ell) \cdot p_i^{\ell} (1 \m p_i)^{k-\ell-1}\Big)\Big] \\
		= & \, \sum_{\ell = \alpha}^{k} \binom{k}{\ell} p_i^{\ell - 2} (1 \m p_i)^{k-\ell-2} \\
		& \cdot \Big(\ell(\ell \m 1)(1 \m p_i)^2 - 2\ell(k \m \ell)p_i(1 \m p_i ) +  (k \m \ell)(k \m \ell \m 1)p_i^2\Big)\\
		= & \,\Big[\!\sum_{\ell = \alpha}^{k} \binom{k}{\ell} p_i^{\ell - 2} (1 \m p_i)^{k-\ell} \ell(\ell \m 1) \Big] \\
		& + \Big[\!\sum_{\ell = \alpha}^{k} \binom{k}{\ell} p_i^{\ell} (1 \m p_i)^{k-\ell-2} (k \m \ell)(k \m \ell \m 1) \Big] \\
		& - 2 \cdot \Big[\!\sum_{\ell = \alpha}^{k} \binom{k}{\ell} p_i^{\ell - 1} (1 \m p_i)^{k-\ell-1} \ell(k \m \ell) \Big]\\
		= & \,\Big[\!\sum_{\ell = \alpha}^{k} \binom{k}{\ell \m 1} p_i^{\ell - 2} (1 \m p_i)^{k-\ell} (\ell \m 1)(k\m \ell \p 1) \Big] \\
		& + \Big[\!\sum_{\ell = \alpha}^{k} \binom{k}{\ell \p 1} p_i^{\ell} (1 \m p_i)^{k-\ell-2} (\ell \p 1)(k \m \ell \m 1) \Big] \\
		& - 2 \cdot \Big[\!\sum_{\ell = \alpha}^{k} \binom{k}{\ell} p_i^{\ell - 1} (1 \m p_i)^{k-\ell-1} \ell(k \m \ell) \Big]\\
		= & \,\Big[\!\sum_{\ell = \alpha-1}^{k} \binom{k}{\ell} p_i^{\ell - 1} (1 \m p_i)^{k-\ell-1} \ell(k \m \ell)\Big] \\
		& + \Big[\!\sum_{\ell = \alpha+1}^{k} \binom{k}{\ell} p_i^{\ell - 1} (1 \m p_i)^{k-\ell-1} \ell(k \m \ell) \Big]\\
		& - 2 \cdot \Big[\!\sum_{\ell = \alpha}^{k} \binom{k}{\ell} p_i^{\ell - 1} (1 \m p_i)^{k-\ell-1} \ell(k \m \ell) \Big]\\
		= & \, \binom{k}{\alpha\m 1} p_i^{\alpha - 2} (1 \m p_i)^{\alpha-1} (\alpha \m 1)\alpha - \binom{k}{\alpha} p_i^{\alpha - 1} (1 \m p_i)^{\alpha-2} (\alpha \m 1)\alpha \\
		= & \, \binom{k}{\alpha} p_i^{\alpha - 2} (1 \m p_i)^{\alpha-2} (\alpha \m 1)\alpha \cdot (1-2p_i) \tag*{\small \textit{since} $\binom{k}{\alpha}=\binom{k}{\alpha\m 1}$}\\
		\leq & \, 0 \quad\text{ for } p_i \geq \tfrac{1}{2}. \tag*{\qedhere} 
	\end{align*}	
\end{proof}

Next, we show that the progress towards consensus in a single round is limited, in particular around the balanced state. Here we utilize two tools. First of all, we employ Lemma \ref{lem:bound_derivative} that bounds the progress around an almost balanced state but only for the "well behaved" case $k=2\alpha \m 1$. Second, we use Lemma \ref{lem:delta_even_odd} to extend this to even Lemma \ref{lem:delta_domination} to extend it to any $k$ and $\alpha$ for which $\tfrac{k}{2} < \alpha \leq k$.

\begin{restatable}{lemma}{lemprogressboundsingleround}
	\label{lem:progress_bound_single_round}
	Let $k \geq 2$ and $\tfrac{k}{2} < \alpha \leq k$ and $S_i \geq \tfrac{n}{2}$ (w.l.o.g.). Then $\Delta_{i+1} > (k \m 1)  \big(S_i \m \tfrac{n}{2}\big) + t \sqrt{n}$ with probability at most $\tfrac{1}{t^2}$, for any $t \geq 1$.
\end{restatable}

\begin{proof}
	We prove the claim for $k = 2\alpha -1$ for $\alpha \geq 2$, and generalize it further below. Since $\delta(p_i)(\tfrac{1}{2}) = 0$ (cf.\ Lemma \ref{lem:delta} or Figure \ref{fig:delta}and \smash{$\frac{\partial\delta(p_i)}{\partial p_i} \leq k$} by Lemma \ref{lem:bound_derivative} we obtain $\E(\Delta_{i+1}) = \delta(p_i)(p_i) \cdot n \leq (k \m 1) n (p_i \m \tfrac{1}{2})$.
	
	To also obtain the claim with the stated probability, let us look at the variance $\sigma^2 := \Var(S_{i+1}) = \Var(\sum_{j=1}^n X_{i+1,j}) $. After round $i$ before round $i \p 1$ we have $\Var(\Delta_{i+1}) = \Var(S_{i+1} \m S_{i}) = \Var(S_{i+1})  = \sigma^2 $, since $S_{i}$ is a constant offset.
	
	The $X_{i+1,j} \in \{0,1\}$ are independent and identically distributed with $\Var(X_{i+1,j}) \leq 1$, thus $\sigma^2 = n \cdot \Var(X_{i+1,j})  \leq n$.
	Using the Chebyshev inequality we obtain 
	\begin{align*}
		\P\Big(\Delta_{i+1} \geq (k \m 1)n (p_i-\tfrac{1}{2}) + t \sqrt{n} \Big) \leq & \,\P\Big(\Delta_{i+1} \geq \E(\Delta_{i+1}) + t \sigma \Big) \\
		\leq & \,\P\Big(|\Delta_{i+1} - \E(\Delta_{i+1})|\geq t \sigma \Big) \!\leq\! \frac{1}{t^2}.
	\end{align*}
	It remains to generalize the claim for $k, \alpha$.	
	Let $\Delta^{k,\alpha}_i$ and $\delta^{k,\alpha}_i$ denote the absolute and relative expected progress in round $i$ for these specific parameters. Let $k = 2 \alpha-1, \alpha \geq 2$ as before. By Lemma \ref{lem:delta_even_odd} we have $\E(\Delta_i^{k, \alpha}) = n \cdot \delta^{k, \alpha}(p_i) = n \cdot \delta^{k-1, \alpha}(p_i) = \E(\Delta_i^{k-1, \alpha})$, hence
	\begin{align*}
		\P\Big(\Delta_{i+1}^{k-1,\alpha} \geq (k \m 1)n (p_i-\tfrac{1}{2}) + t \sqrt{n} \Big) \leq & \,\P\Big(\Delta_{i+1}^{k-1,\alpha} \geq \E(\Delta_{i+1}^{\text{\boldmath$k$},\alpha)} + t \sigma \Big) \\
		= & \,\P\Big(\Delta_{i+1}^{k-1,\alpha} \geq \E(\Delta_{i+1}^{\text{\boldmath$k\m1$},\alpha)} + t \sigma \Big) \!\leq\! \frac{1}{t^2}.
	\end{align*}
	Note that this extends the claim from $k = 2\alpha \m1$ (odd) to the case $k = 2\alpha \m 2$ (even).
	
	Let now $k$ be arbitrary (odd or even) and $\tfrac{k}{2} < \alpha' \leq k$. By Lemma \ref{lem:delta_domination} we have that \[\E(\Delta_i^{k, \alpha'}) = n \cdot \delta(p_i)^{k, \alpha'} \leq n \cdot \delta^{k, \alpha}(p_i) = \E(\Delta_i^{k, \alpha}).\] Applying the Chebyshev bound once more gives us
	\begin{align*}
		\P\Big(\Delta_{i+1}^{k,\alpha'} \geq (k \m 1)n (p_i-\tfrac{1}{2}) + t \sqrt{n} \Big) \leq & \,\P\Big(\Delta_{i+1}^{k,\alpha'} \geq \E(\Delta_{i+1}^{k,\text{\boldmath$\alpha$}}) + t \sigma \Big) \\
		\leq & \,\P\Big(\Delta_{i+1}^{k,\alpha'} \geq \E(\Delta_{i+1}^{k,\text{\boldmath$\alpha'$}}) + t \sigma \Big) \!\leq\! \frac{1}{t^2}.\tag*{\qedhere}
	\end{align*}	
\end{proof}

Building on the previous lemma, we can give the following probabilistic bound for the number of parties that have opinion 1 after $i$ rounds.

\begin{restatable}{lemma}{lemprogressboundmultiround}
	\label{lem:progress_bound_multi_round}
	Let $k \geq 2$, $\tfrac{k}{2} < \alpha \leq k$. Assume the system is in a roughly balanced state with $S_0 \leq \tfrac{n}{2} + f(n)$ for  \smash{$f(n) = \sqrt{n \log n}$}. Then for any \smash{$i \leq \frac{\log n}{c}$} it holds $S_{i} > \frac{n}{2} + (k \p 1)^i f(n)$ with probability at most $1/c$ for any $c \geq 1$.
\end{restatable}

\begin{proof}
	Let us first assume that each round the statement $\Delta_{i+1} \leq  (k \m 1)  \big(S_{i} \m \tfrac{n}{2}\big) + f(n)$ is true and assess the overall progress we make at most towards a 1 consensus starting from $S_0$. 
	\[
	S_1 = \Delta_1 + S_{0} \leq (k \m 1)  \big(S_{0} \m \tfrac{n}{2}\big) + f(n) + S_0 \leq \tfrac{n}{2} + (k \m 1)  f(n) + 2f(n) = \tfrac{n}{2} + f(n) (k \p 1).
	\]
	This satisfies the lemma for the first round and serves as our induction base. For the induction step we obtain
	\begin{align*}
		S_{i+1} = & \,\Delta_{i+1} + S_{i} \leq  (k \m 1)  \big(S_{i} \m \tfrac{n}{2}\big) + f(n) +  S_i \tag*{\small\textit{(Lemma \ref{lem:progress_bound_single_round})}}\\ 
		\leq & \,(k \m 1)  (k\p 1)^i  f(n)  + f(n) + \tfrac{n}{2} +  (k\p 1)^i  f(n)\tag*{\small\textit{(induction hypothesis)}}\\
		= & \, \tfrac{n}{2} + k  (k\p 1)^i  f(n)  + f(n) \\
		\leq & \,\tfrac{n}{2} +  (k\p 1)^{i + 1} f(n) 
	\end{align*}
	
	To conclude the proof, we compute the probability that the claim is true. By Lemma \ref{lem:progress_bound_single_round} we have that $\Delta_i > (k \m 1)  \big(S_i \m \tfrac{n}{2}\big) + \sqrt{n \log n}$ with probability at most \smash{$\tfrac{1}{\log n}$}.  Union bounding this for at most \smash{$\tfrac{\log n}{c}$} rounds, we obtain that the claim is true with probability at most \smash{$\tfrac{\log n}{c}\tfrac{1}{\log n}= \tfrac{1}{c}$}.
\end{proof}

We have all tools to deduce the lower bound for the number of rounds of Slush that is required even to get moderately close to a consensus state with some moderate probability.

\begin{restatable}{theorem}{thmslushlowerbound}
	\label{thm:slush_lower_bound}
	For $k \geq 2$ and any $\tfrac{k}{2} < \alpha \leq k$ and sufficiently large $n$, running Slush for at most \smash{$ \tfrac{\log n}{3\log (k+1)\log \gamma}$} rounds there is a majority opinion with at least $\tfrac{n}{2} + \tfrac{n}{\gamma}$ parties with probability at most \smash{$\tfrac{1}{\log(k+1)\log \gamma}$} for any constant $\gamma \geq 2$.
\end{restatable}

\begin{proof}
	We will assume that there are initially no parties with opinion $\bot$ (cf.\ the explanation of Slush in Section \ref{sec:protocols}), which only strengthens the lower bound. Furthermore, we restrict ourselves to bound the probability of a 1-majority, since the same claim for a 0-majority holds by symmetry (and then we apply a union bound for the probability of one of either consensus happening). 
	
	The scenario for our lower bound for a 1 majority is the start state $S_0 \leq \tfrac{n}{2} + f(n)$ for \smash{$f(n) = \sqrt{n \log n}$} with $n > 16$, which conforms to the preconditions of Lemma~\ref{lem:progress_bound_multi_round}.
	We choose the parameter $c$ from Lemma \ref{lem:progress_bound_multi_round} as \smash{$c=   3 \log (k\p 1) \log \gamma $}. Then, by Lemma \ref{lem:progress_bound_multi_round}, after at most \smash{$i \leq \frac{\log n}{c}$} rounds it is \smash{$S_ i > \frac{n}{2} + (k \p 1)^i f(n)$} with probability at most $\tfrac{1}{c}$. Furthermore we have
	\[
	(k \p 1)^i f(n) \leq f(n) \cdot (k \p 1)^{\frac{\log n}{3\log (k+1) \log \gamma}} = f(n) \cdot (n/\gamma)^{1/3} \leq \frac{n}{\gamma}.
	\]	
	In the last step we use that $f(n) \leq (n/\gamma)^{2/3}$ for sufficiently large $n$. Note that the same holds for obtaining a majority in the opinion 0 by starting in a state $\tfrac{n}{2} - f(n) \leq S_0 \leq \tfrac{n}{2} + f(n)$ due to symmetry. We obtain the claim of the theorem (for $k=2\alpha-1$) through a union bound on the number of rounds to either a 0-consensus or a 1-consensus with $\tfrac{2}{c} \leq \tfrac{1}{\log(k+1)\log \gamma}$.
\end{proof}

We express the theorem above in a simpler, albeit weaker form.

\begin{corollary}
	\label{cor:slush_lower_bound}
	For $k \geq 2$ and any $\tfrac{k}{2} < \alpha \leq k$, Slush takes \smash{$\Omega\big( \tfrac{\log n}{\log k}\big)$} rounds in expectation to reach a stable consensus (as defined in Definition \ref{def:consensus_state}).
\end{corollary}

\subsection{Upper Bound}

We show how to use the structural insights about Slush with respect to parameters $k$ and $\alpha$ to extend known upper bounds for the so called Median protocol, the 3-Majority protocol and the 2-Choices protocol. These are usually conceptualized for the case of multiple ($>2$) opinions and are defined as follows.

\begin{definition}[cf.\ \cite{Becchetti2020}]
	\label{def:randomized_consensus_protocols}
	The \emph{Median protocol} assumes some globally known total order among opinions. In each round, each party samples the opinion of two others and adopts the median among those two and its own.
	
	In the \emph{3-Majority} protocol, in each round, each party samples the opinion of three others and adopts the majority opinion, or picking a random opinion among the three in case of a tie.
	
	In the \emph{2-Choices} protocol, in each round, each party samples the opinion of two others and then applies the 3-Majority rule, defaulting to its own opinion in case of a tie.	
\end{definition}

We make the following observation.

\begin{remark}
	\label{obs:equivalence_slush}		
	Assume that all parties have initially only one of two opinions (i.e., the binary case, in particular, there are no parties with opinion $\bot$). Then the Median protocol, 2-choices protocol and Slush for $k=2$ and $\alpha = 2$ are all equivalent. This is because in the \emph{binary case}, in all three protocols a given party will switch its own opinion if and only if it samples two parties that both have a different opinion from its own. 	
	Under the same circumstances and for the same reason, the 3-Majority protocol is equivalent to Slush for $k=3$ and $\alpha = 2$ (exploiting that there can never be a tie in the binary case).
\end{remark}

There has been extensive research on the dynamics of the Median, 2-Choices and 3-Majority protocol (Definition \ref{def:randomized_consensus_protocols}) and the techniques are for the most part analogous or at least quite similar if the number of opinions is kept constant. We have already established the lower bound of \smash{$\Omega\big(\tfrac{\log n}{\log k}\big)$}, i.e., the \emph{additional} speed-up one can gain by increasing the query size $k$ diminishes very fast. Therefore, we do not deem it particularly worthwhile to show an upper bound that strictly improves on the $\bigO(\log n)$ bound for the aforementioned cases.

Furthermore, it is \textit{not} the scope of this paper to give detailed proofs of slight generalizations of those for the protocols from Definition \ref{def:randomized_consensus_protocols}. To keep this paper reasonably self-contained we showcase how these proofs generalize to Slush with arbitrary $k\geq 2$ and $\alpha = \lceil \tfrac{k+1}{2}\rceil$. We will give an extended proof sketch that shows how the existing proof techniques generalize to obtain the following theorem. For more details we refer to the according sources (in particular \cite{Doerr2011,Becchetti2017}).

\begin{restatable}{theorem}{thmslushupperbound}
	\label{thm:slush_upper_bound}
	Let $k\geq 2$ and $\alpha = \lceil \tfrac{k+1}{2}\rceil$. Then Slush reaches a state where all but $n-O(\!\sqrt n)$ have the same opinion in $\bigO(\log n)$ rounds with high probability, even in the presence of a \smash{$\sqrt{n}$}-bounded adversary.	
\end{restatable}

\begin{proof}[Proof Sketch]
	Our proof rests on the structure of the according proof for the $k=\alpha=2$ and $(k,\alpha)=(3,2)$ from \cite{Doerr2011,Becchetti2017} for the binary case. The proof is divided into a constant number of phases. The phases range from the worst case where the distribution of opinions is roughly in an equilibrium, to the state of a stable consensus.
	We will show that it always takes at most $\bigO(\log n)$ rounds to arrive in the corresponding next phase, even when starting from the worst case of an equilibrium.
	Due to symmetry we assume $S_i \geq 0$ w.l.o.g. 
	We will at first make the main argument without considering the \smash{$\sqrt{n}$}-bounded adversary and the opinion $\bot$ (see Section \ref{sec:slush}) and argue why the proof holds for these cases at the end.
	
	\textbf{Phase 1: \smash{\boldmath$\tfrac{n}{2} \leq S_i \leq \tfrac{n}{2}  + c_1\sqrt{n \log n}$}.} To lift $S_i$ over the threshold of \smash{$\tfrac{n}{2}  + c_1\sqrt{n \log n}$} to the next phase, anti-concentration bounds are used. In particular, because $\alpha = \lceil \tfrac{k+1}{2}\rceil$ there is always a majority for either 0 or 1 in each query, thus we have that $X_{ij} \sim \mathcal B(q_i)$ is Bernoulli distributed with probability $q_i \approx p_i \approx \tfrac{1}{2}$ (since the system is close to balanced state). Hence, \smash{$S_i = \sum_{j=1}^{n} X_{ij}$} is a sum of independent, identically distributed Bernoulli variables. By the Central Limit Theorem, for sufficiently large $n$ the variable $S_i$ approximates a normal distribution with deviation of $\Omega\big(\!\sqrt{n}\big)$ around the expectation $\E(S_i) \approx \tfrac{n}{2}$. The property of the normal distribution implies that there is a constant probability for $|S_i - \tfrac{n}{2}|\geq   c\sqrt{n}$ in a single round. 
	Applying concentration (Chernoff) bounds and considering that by symmetry we are allowed to escape Phase 1 in either direction, one can ensure $|S_i -\tfrac{n}{2}| > c_1\sqrt{n \log n}$ after $\bigO(\log n)$ rounds w.h.p.\ (cf. \cite{Doerr2011} for more details).
	
	\textbf{Phase 2: \smash{\boldmath$\tfrac{n}{2}  + c_1\sqrt{n \log n} < S_i \leq \tfrac{n}{2}  +\tfrac{n}{c_2}$}.}
	We explain the idea in a style that is akin to the proof by \cite{Becchetti2017}, which shows the argument for $(k,\alpha)=(3,2)$ and then extend it to the general case.  Note that expected progress is  $\E(\Delta_i) = n \cdot \delta^{3,2}(p_i)$. The next step is to show that in the Phase 2 interval of $S_i$ (here $c_2=3$ in general, $c_2$ depends on $k, \alpha$) the function $\delta^{3,2}(p_i)$ can be lower bounded by a linear function of constant positive gradient through the point $(\tfrac{1}{2},0)$ (see Figure \ref{fig:delta_a} for a visual representation). This implies that in each round the expected progress increases linearly in the progress that was made in the previous round.
	Intuitively, this corresponds to a situation of ``compounding interest'' in expected progress $\E(\Delta_i)$ with each round $i$, i.e., exponential growth in $i$. Hence, it takes at most $\bigO(\log n)$ rounds until $\Psi_i > \tfrac{n}{c_2}$. 
	
	The caveat is, that this only works if we can guarantee that the system makes progress that is at least some constant fraction of the expected progress. Due to randomness the expected progress could be undershot or the system could even backslide towards the equilibrium. To show that this does not happen w.h.p., we exploit the intuition that in Phase 2 the system is already relatively far advanced into a majority, such that, the expected progress $\E(\Delta_i)$ exceeds the standard deviation of $\Delta_i$. Concretely, one can use concentration bounds and union bounds, to guarantee a constant fraction of the expected progress $\E(\Delta_i)$ w.h.p.\ for each round in this phase.
	
	By Lemma \ref{lem:delta_even_odd} we have that $\delta^{3,2}(p_i) = \delta^{2,2}(p_i)$, so the argument also extends to the 2-Choices protocol. Finally, Lemma \ref{lem:delta_domination2} shows that the expected progress $\delta^{k,\alpha}$ for $k \geq 4$ and $\alpha = \lceil \tfrac{k+1}{2}\rceil$ dominates that for $k \in \{2,3\}$, and consequently the argument above applies for Slush in general.
	
	\textbf{Phase 3: \smash{\boldmath$\tfrac{n}{2}  + \tfrac{n}{c_2} < S_i \leq \tfrac{n}{2}  +\tfrac{n}{c_3}$.}} This is arguably the simplest phase, since the system is a constant fraction of parties from both the equilibrium state and the consensus state. This implies that $\delta^{2,2}(p_i)$ ($= \delta^{3,2}(p_i)$ by Lemma \ref{lem:delta_even_odd}) can be lower bounded by a constant (cf.\ Figure \ref{fig:delta_a}), therefore a constant fraction of all parties will switch to the majority opinion in expectation. It is not hard to show that this also holds w.h.p. and by Lemma \ref{lem:delta_domination2}, also for any $k \geq 4$ and $\alpha = \lceil \tfrac{k+1}{2}\rceil$.
	
	\textbf{Phase 4: \smash{\boldmath$\tfrac{n}{2}  + \tfrac{n}{c_3} < S_i \leq n - c_4 \sqrt{n}$.}} Although the expected progress $\delta^{2,2}(p_i)$ tends to 0 as $p_i$ approaches 1 (see Figure \ref{fig:delta_a}), one can apply a similar argument as in Phase 2 but ``backwards''. In particular, in the corresponding interval of $S_i$ the expected progress given by $\delta^{2,2}(p_i)$ can be lower bounded by a linear function with constant negative slope through the point $(1,0)$ (see Figure \ref{fig:delta_a} for a visual confirmation). This implies that the expected number of parties holding the minority opinion shrinks by a constant fraction in each round (cf.\ \cite{Becchetti2017}). Thus it takes at most $\bigO(\log n)$ rounds until $S_i$ passes the threshold $n - c_4\sqrt{n}$ given that the progress is at least a constant fraction of the expected progress. As in Phase 2, the latter can be guaranteed w.h.p.\ using concentration bounds. The argument generalizes to $k \geq 3$ and $\alpha = \lceil \tfrac{k+1}{2}\rceil$ by Lemmas \ref{lem:delta_even_odd} and \ref{lem:delta_domination2}. 
	
	\textbf{\smash{\boldmath$\sqrt{n}$}-Bounded Adversary:} The optimal strategy of the adversary to avoid a consensus is to flip $\sqrt{n}$ parties of the majority to the minority opinion. The main argument is that the ability of adversary to influence $S_i$ is asymptotically not more than the standard deviation of $S_i$. In Phase 1 the random deviation from the expectation is a desired effect to lift the system out of an equilibrium and one can show that w.h.p., the given adversary can not inhibit this. In the subsequent phases, the random deviation is undesired as it may reduce the progress below its expectation. The intuition is that since the influence of the adversary is actually less than the standard deviation, we can deal with both the standard deviation and the effect of the adversary as before, by adjusting constants in the running time.
	
	\textbf{Dealing with parties with opinion \boldmath$\bot$:} It remains to argue that the asymptotic running time does not change if we introduce the special opinion $\bot$, which is relatively straight forward. Whenever a party $j_1$ with opinion $\bot$ receives a query for an opinion by another party $j_2$ then $j_1$ adopts $b$ (see Algorithm \ref{algo:slush}).	
	Let $n_{0,1}$ and $n_{\bot}$ be the number of parties that have opinion $0,1$ or $\bot$ respectively ($n = n_{0,1} + n_{\bot}$). As long as $n_{0,1} \leq n_{\bot}$ it can be shown that $n_{0,1}$ doubles every $\bigO(1)$ rounds w.h.p. Similarly, if $n_{0,1} \geq n_{\bot}$, the number of parties halves every $\bigO(1)$ rounds w.h.p. Ultimately, this implies that the opinion $\bot$ will die out after $\bigO(\log n)$ rounds.
\end{proof}

We can translate the above result into the notion of concentration with all but negligible probability (Definition \ref{def:all_but_neglible_probability}) by adding a factor of $\beta$ to the running time that gives more control over the level of security in in particular for small $n$ (see Lemma \ref{lem:whp_abn} and Remark \ref{rem:whp_abn} for the details).  Specifically, the corollary conforms to Definition \ref{def:negligible}, as the runtime is polynomial in $\beta$.

\begin{corollary}
	\label{cor:slush_upper_bound}
	Let $k\geq 2$ and $\alpha = \lceil \tfrac{k+1}{2}\rceil$. Then Slush reaches a stable consensus  in $\bigO(\log n + \beta)$ rounds with all but negligible probability (with respect to $\beta$), even in the presence of a \smash{$\sqrt{n}$}-bounded adversary.	
\end{corollary}

\section{Dynamics of Snowflake and Snowball}
\label{sec:snowball_dynamics}

In this section we are going to extend the lower bound for Slush derived in Section \ref{sec:slush_bounds} to the Snowflake and Snowball protocols.
Note that the quantities $S_i, p_i, \Delta_i, \delta_i$ can be defined the same as in Slush, see Section \ref{sec:slush_dynamics}, since the variables only depend on the opinion attribute of parties, which is present in all three protocols. However, the actual (expected) changes of these quantities in this section can and will differ from those in Slush.
We will denote these quantities with superscripts (slush, flake, ball) in case we compare them over protocols (but avoid this whenever possible). Moreover, we condition the results of this section on the assumption that no node decides (finalizes) their opinion, before the system reaches a stable majority and consider the repercussions at the end.

Snowball, as explained in Section~\ref{sec:snowball}, augments the consensus mechanism from Slush with the concept of \emph{confidence} associated to the current value, which influences the decision of a party to change its opinion. In a nutshell, a node changes opinion in the Snowball protocol when the cumulative number of queries with majority for the new opinion exceeds that for the old opinion. 
In the Slush protocol, the variable $S_i$ was sufficient to describe the expected progress required to predict the evolution of the system, which is not the case in Snowball anymore, since aforementioned confidence levels play a crucial role.

\begin{definition}
	Define the set $L_i^{c}$ to be the set of parties in round $i$ such that $\op{cnt}(1)-\op{cnt}(0)=c$ (where $\op{cnt}(b)$ is the number of queries of a given party that had a majority of opinion $b$). We further divide the set $L_i^0$ in two subsets $L_i^{0,0}$ and $L_i^{0,1}$. Parties in $L_i^0$ ($L_i^{0,1}$) that have opinion $0$ ($1$) belong to $L_i^{0,0}$ ($L_i^{0,1}$).
\end{definition}

	The variable $S_i$ can be reconstructed as follows: \smash{$S_i=|L_i^{0,1}|\p\sum_{c>0}|L_i^c|$}.
	Given a round $i>0$, the set of parties \CN is contained in $\bigcup_{c=-i}^i L_i^{c}$ as the end of round $i$, 
	since every party performed $i$ queries by the end of round $i$, thus the difference in counts $c$ is bounded between $-i$ and $i$.

\begin{remark}
	\label{remark:evo}
	Consider the collection $\CL_i:=\{L_i^{c}\}_{c=-i}^i \cup \{L_i^{0,0},L_i^{0,1}\}$ of disjoint sets. The evolution of the system in the next round $i+1$ can now be described using this set $\CL_i$.
	After a query is performed a party in $L_{i}^c$ moves to set $L_{i}^{c+1}$ if the query had a majority for $1$, to $L_{i}^{c-1}$ if the query had a majority for $0$, or $L_{i}^{c}$ if the query had no majority.
	The only parties that can change value after round $i$ are the parties contained in $L_i^{0}$.
\end{remark}

\begin{definition}
	\label{def:sbdelta}
	For $i\geq 1$, we define the absolute progress as $\Delta_i:=S_i-S_{i-1}=|L_{i}^{0,1}|- |L_{i-1}^{0,1}|+\sum_{c=1}^{i}(|L_i^{c}|-|L_{i-1}^{c}|)$, i.e., the number of parties with $1$ in their view in round $i$ minus the number of parties with $1$ in their view in round $i-1$. As before, we define the expected relative progress as $\delta_i:=\E(\Delta_i)/n$.
\end{definition}

The following Lemma shows that in Snowball, $\Delta_i$ is only affected by parties migrating from $L_{i-1}^{0,1}$ or $L_{i-1}^{0,0}$ in round $i\m1$ to $L_i^{-1}$ or $L_i^{1}$ in round $i$, respectively.

\begin{restatable}{lemma}{lemmasbprogress}
	\label{lemma:sbprogress}
	The absolute progress can be expressed as $\Delta_i=|\Lambda_i^1|-|\Lambda_i^0|$, where
	$\Lambda_i^{0} := L_i^{-1}\cap L_{i-1}^{0,1}$ and $\Lambda_i^{1} := L_i^{1}\cap L_{i-1}^{0,0}$.
\end{restatable}	

\begin{proof}
	Recall $\Delta_i=|L_{i}^{0,1}|- |L_{i-1}^{0,1}|+\sum_{c=1}^{i}(|L_i^{c}|-|L_{i-1}^{c}|)$ and consider a party $j$. We distinguish the following cases:
	\begin{itemize}
		\item If $j\in L_{i-1}^c$ for $c<0$, then Remark~\ref{remark:evo} guarantees that $j\in L_{i}^{c-1} \cup L_{i}^{c}\cup L_{i}^{c+1}\cup L_i^{0,0}$. None of the previous sets are involved in the definition of $\Delta_i$, thus the value of $\Delta_i$ is independent from party $j$.
		\item If $j\in L_{i-1}^c$ for $c>0$, then Remark~\ref{remark:evo} guarantees that $j\in L_{i}^{c-1} \cup L_{i}^{c}\cup L_{i}^{c+1}\cup L_i^{0,1}$. In the definition of $\Delta_i$, the terms $|L_i^{c^\prime}|$ for $c^\prime\in\{c-1,c,c+1, \{0,1\}\}$ and $|L_{i-1}^{c}|$ appear with opposite signs. Thus, the contribution of party $j$ cancels out.  
		\item If $j\in L_{i-1}^{0,0}$, then Remark~\ref{remark:evo} guarantees that $j\in L_{i}^{-1} \cup L_{i}^{0,0}\cup L_{i}^{1}$. If $j\in L_{i}^{-1}$ or $j\in L_{i}^{0,0}$, the terms terms $L_{i-1}^{0,0}$, $j\in L_{i}^{0,0}$, and $ L_{i}^{-1}$ do not appear in the definition of $\Delta_i$. If $j\in L_{i}^{1}$, then $j\in L_{i-1}^{0,0}\cap L_{i}^{1}=\Lambda_i^1$ and $j$ contributes with $+1$ to $\Delta_i$.
		\item If $j\in L_{i-1}^{0,1}$, then Remark~\ref{remark:evo} guarantees that $j\in L_{i}^{-1} \cup L_{i}^{0,0}\cup L_{i}^{1}$. If $j\in L_{i}^{1}$ or $j\in L_{i}^{0,1}$, the term $ L_{i-1}^{0,1}$ appears with coefficient $-1$, whereas the terms $j\in L_{i}^{1}$ and $j\in L_{i}^{0,1}$ appear with coefficient $1$. Thus, their contribution cancel out. If $j\in L_i^{-1}$, then $j\in L_{i-1}^{0,1}\cap L_i^{-1}=\Lambda_i^0$ the term $j\in L_{i-1}^{0,1}$ appears with coefficient $-1$, whereas the term $j\in L_i^{-1}$ does not appear. 
	\end{itemize}
	We conclude that $\Delta_i=|\Lambda_i^1|-|\Lambda_i^0|$.
\end{proof}

An interesting interpretation of Lemma~\ref{lemma:sbprogress} if the following. Since the parties contained in the set $L_i^0$ are the only parties that can change their opinion,  the expected progress of Snowball in a given round is the same as the expected progress of Slush restricted to the parties in $L_i^0$. We formalize this intuition in the following lemma.

\begin{restatable}{lemma}{lemmacomparison}
	\label{lemma:comparison}
	The expected absolute progress of Snowball is at most as high as in Slush, i.e., $\delta^\var{ball}_i\leq\delta^\var{slush}_i$.
\end{restatable}

\begin{proof}
	Lemma~\ref{lemma:sbprogress} states that only parties in $L_{i-1}^0$ can modify the state $S_i$. Given a party $j\in L_{i-1}^{0,0}$ (respectively $j\in L_{i-1}^{0,1}$)
	The probability that $j$ changes to value $1$ (respectively $0$) is given by 
	\begin{align*}
		\P\big(X_{ij} & = 1 \mid j\in L_{i-1}^{0,0}\big) = \P\big(Y_{ij} \geq \alpha \big)= \sum_{\ell = \alpha}^{k} \binom{k}{\ell} p_i^\ell (1 \m p_i)^{k-\ell}.\\
		\P\big(X_{ij} & = 0 \mid j\in L_{i-1}^{0,1}\big) = \P\big(Y_{ij} \leq k-\alpha \big) 
		= \sum_{\ell = 0}^{k-\alpha} \binom{k}{\ell} p_i^\ell (1 \m p_i)^{k-\ell},
	\end{align*}
	where the quantities $X_{ij}, Y_{ij}$ are defined the same as in Section \ref{sec:slush_dynamics}, i.e., we deal with exactly the same probabilities as in Lemma~\ref{lem:delta}. 
	Consequently, we are able to apply Lemma \ref{lem:delta} and conclude that the expected rate of progress is $\delta^{\text{ball}}_i = \delta(p_i)$ when restricted to $L_{i-1}^0$.
	
	The latter is an important caveat, since what changes in Snowball is the set of parties on which the expected rate of progress is applied. Then the way the expected absolute progress of Snowball and Snowflake relate to each other is given as follows
	\begin{equation*}
		\E(\Delta^\text{ball}_i)
		= |L_i^0|\cdot\delta^\text{ball}_i
		= |L_i^0|\cdot\delta(p_i)
		= |L_i^0|\cdot\delta^\text{slush}_i
		= \tfrac{|L_i^0|}{n}\cdot\E(\Delta^\text{slush}_i)
		\stackrel{|L_i^0|\leq n}{\leq}\E(\Delta^\text{slush}_i).
	\end{equation*}
	Dividing by $n$ on both sides and using $\delta_i^{\text{ball}} = \frac{\E(\Delta^\text{ball}_i)}{n}$, $\delta_i^{\text{slush}} = \frac{\E(\Delta^\text{slush}_i)}{n}$ (see start of Section \ref{sec:slush_dynamics}), yields the desired result.	
\end{proof}	

Lemma~\ref{lemma:comparison} states that the expected progress of the Snowball protocol is upper-bounded by the expected progress of the Slush protocol. We conclude that the expected number of rounds even to reach majority of a constant fraction of nodes of one opinion of the Snowball protocol is lower-bounded by the Slush protocol, if no node decides prematurely. Note that the same is clearly true for as Snowflake which is essentially equal to Slush if  no node decides prematurely.

\begin{corollary}[cf.\ Corollary \ref{cor:slush_lower_bound}]
	\label{cor:ball_lower_bound}
	For $k \geq 2$ and any $\tfrac{k}{2} < \alpha \leq k$, Snowball and Snowball take \smash{$\Omega\big( \tfrac{\log n}{\log k}\big)$} rounds in expectation to reach a state state of stable consensus for any constant $\gamma \geq 2$, assuming that nodes do not decide before such a state is reached.
\end{corollary}

Note that since the decision mechanism in Snowflake and Snowball implies that no node can decide before $\beta$ rounds have passed the corollary implies a lower bound of \smash{$\Omega\big(\min\big( \tfrac{\log n}{\log k}, \beta\big)\big)$} rounds. Furthermore, we will see in Section \ref{sec:security_ball_flake} that the dependence of the runtime on $\beta$ behaves much worse than $\Omega(\beta)$ as the adversary can exploit the decision mechanism to delay a decision \emph{super-polynomially} in $\beta$.

\section{Security of Snowflake and Snowball}
\label{sec:security_ball_flake}

We show that in Snowflake and Snowball, has a vulnerability towards an adversary that intends to delay consensus (as defined in Definition~\ref{def:consensus}). In particular, there might is an unfavorable trade-off between confidence of success with and latency.
In particular, the mechanic that Snowflake and Snowball protocols use introduces a security parameter $\beta$ to control the probability of failure of obtaining a consensus (according to Definition \ref{def:consensus}). 
We show that this mechanism to make decision allows an adversary to delay the decision of any given party when using the consensus mechanisms of Snowflake and Snowball for a super-polynomial number of rounds in $\beta$.
This is independent of the current state of the system, i.e., the claim is true even if the system is in a state of a stable consensus (see Definition \ref{def:consensus_state}) and is true for a weaker notion of the $F$-bounded adversary from Definition \ref{def:adversary} for a small $F$.

\begin{definition}
	\label{def:weak_bounded_adversary}
	A weak $F$-bounded adversary controls up to $F$ \emph{undecided} parties whose state (opinion) it can set once each round. We call these \emph{influenced parties} and in particular we assume that the adversary can reset any decision on some opinion made by those.
\end{definition}

We start by giving a lower bound for the probability that some party samples a majority of influenced parties. 

\begin{restatable}{lemma}{lemsinglepartyprobability}
	\label{lem:single_party_probability}
	The probability that a random sample of $k$ parties contains at least $\alpha$ that are influenced by a weak $F$-bounded adversary is at least $\big(\frac{F}{n}\big)^k$.
\end{restatable}

\begin{proof}
	Recall that we nodes are sampled uniformly at random with repetition.
	We can lower bound $q_j$ with the probability that $j$ samples at least $\alpha$ Byzantine parties as follows
	\[
	q_j = \sum_{\ell = \alpha}^{k} \binom{k}{\ell}\Big(\frac{F}{n}\Big)^\ell \Big(\frac{n-F}{n}\Big)^{k-\ell} \geq \sum_{\ell = k}^{k} \binom{k}{\ell}\Big(\frac{F}{n}\Big)^\ell \Big(\frac{n-F}{n}\Big)^{k-\ell} = \Big(\frac{F}{n}\Big)^k.\tag*{\qedhere}
	\]
\end{proof}

Note that even though the probability above decreases with $k$, the parameter $k$ is considered a small, constant sized tuning parameter \cite{Rocket2019} and is not a proper security parameter, particularly since the message complexity scales in $\Omega(kn)$.

Interestingly, the lemma shows that even in a stable consensus (according to Definition \ref{def:consensus_state}), i.e., where almost all nodes share the same opinion, even a weak adversary can create a small but inherent ``background noise'', i.e., an expected fraction \smash{$\big(\frac{F}{n}\big)^k$} of all parties can be reverted by the adversary to the minority opinion each round, because it gains a majority in a sample.

This situation is what the security mechanic of Snowflake and Snowball protocols is intended for as it makes parties decide and finalize an opinion by introducing an according mechanism with a security parameter $\beta$. One condition for some party to decide an opinion, is that it must have at least $\beta$ consecutive queries with an $\alpha$-majority of the same opinion, see Section \ref{sec:protocols} or Section \ref{apx:pseudocodes} for detailed pseudocode. (Note that Snowball imposes an additional condition for parties to decide based on the history of queries, which, however, only delays the decision of a party down even further).

The idea behind this mechanism is to reduce the probability that an adversary can make some party accept the \emph{minority} opinion, since sampling the minority opinion $\beta$ times in a row has a probability that is negligible with respect to $\beta$.  This mechanism is flawed in the sense that even an adversary that influences just a single party can abuse it to introduce a delay to the decision of a party that scales badly in $\beta$.
 
\begin{restatable}{lemma}{lemminqueries}
	\label{lem:min_queries}
	In the Snowflake and Snowball protocols, there exists a value $c > 1$ which is constant in $\beta$, such that a weak $F$-bounded adversary (Def.\ \ref{def:weak_bounded_adversary}, for any $F \geq 2$) can ensure that the probability that a given party decides within $c^{\beta-1}/2$ queries (rounds) is at most $4/c^{2(\beta-1)}$.
	To enforce this, the adversary needs no information on the current state of the network. 
\end{restatable}

\begin{proof}
	In the Snowflake and Snowball protocols,  a variable $cnt \geq 1$ maintains the length of the most recent sequence of consecutive queries which all had a majority of the same opinion $b$ (see Section \ref{apx:pseudocodes}). To decide on some opinion $b$ (which we call ``success'' in the following), it is necessary that $cnt \geq \beta$. 
	
	Consider the following adversary strategy, where it splits its influenced parties into two roughly equally sized groups  whose opinions it sets to 0 and 1, respectively (here we need $F\geq 2$, round group sizes if necessary).
	We will now compute the probability $q_\beta$ of the event $E_\beta$ (failure in the sequence of length $\beta$) that some given sequence of $\beta$ queries of some party $j$ contains either a query which had no opinion or at least one of each opinion, 0 or 1.
	
	Assume, that the first opinion in such a sequence had a majority of 1. This is w.l.o.g., firstly, since the case of no majority in the first query already satisfies the condition above (by the law of total probability, neglecting this case gives us a lower bound for $q_\beta$). Secondly, because the case where a sequence starts with a 0 majority is analogous, if we treat the two groups of the adversary as two separate adversaries of size at least $F' \geq \lfloor \tfrac{F}{2} \rfloor \geq 1$.
	
	This means $q_{\beta}$ is lower bounded by the probability that the next $\beta-1$ queries have a majority from the $F'$-bounded adversary. Thus $q_\beta \geq (1-q_j)^{\beta-1}$, where $q_j$ is the lower bound of the probability that at least $\alpha$ Byzantine parties are sampled in a given query from Lemma \ref{lem:single_party_probability}. Note that $q_j$ is lower bounded by a non-zero value that is independent of $\beta$.

	The number of queries crucially depends on the probability $q_\beta$ of failure of deciding in a sequence of length $\beta$. To measure the number of queries until success we introduce a r.v.\ $Z_j$ that lower bounds the number counter resets  because event $E_\beta$ occurs (with probability $q_\beta$). This means that $Z_j$ follows a geometric distribution $Z_j \sim \mathcal G(q_\beta)$, describing the number of failures (i.e., counter resets) until success (deciding on a value).
	
	Note that $Z_j$ is a conservative lower bound for the overall number of queries that $j$ has to make. First, we significantly underestimate the probability of failure $q_\beta$. Second, we consider only counter resets caused by sampling a Byzantine majority. Third, we do not account for queries (that increase the counter) in between counter resets.
	
	The expectation of $Z_j \sim \mathcal G(q_\beta)$ is $\mu = 1/q_\beta$ and the variance is $\sigma := (1-q_\beta)/q_\beta^2$. Then the probability that $X$ is at most half its mean can be bounded with the Chebyshev inequality:
	\begin{align*}
		\P(X \leq \tfrac{\mu}{2}) & \leq \P( |X - \mu| \geq \tfrac{\mu}{2})\\
		& =  \P(|X - \mu| \geq k \cdot \sigma) \tag*{\textit{where} $k := \tfrac{\mu}{2\sigma}$}\\
		& \leq \tfrac{1}{k^2} = \smash{\tfrac{4(1-q_\beta)^2}{q_\beta^2}} \leq \tfrac{4}{q_\beta^2}.
	\end{align*}
	Let $c := (1-p_j)^{-1}$ ($> 1$), thus $q_\beta \geq c^{-(\beta-1)}$. The claim follows from the fact that $\mu = 1/q_\beta = c^{\beta-1}$ and $\P(X > \tfrac{\mu}{2}) = 1 - \P(X \leq \tfrac{\mu}{2}) \geq 1-{4}/{c^{2(\beta-1)}}$.	
\end{proof}

Using the lemma above, we show that Snowflake and Snowball cannot satisfy Definition \ref{def:negligible}, which states the conditions for a mechanism that provides a decent trade-off between security and performance. Specifically, the following theorem shows that having consensus with all but negligible probability w.r.t., $\beta$ and a polynomial runtime in $\beta$ are mutually exclusive.

\begin{restatable}{theorem}{thmimpossibility}
	\label{thm:impossibility}
	In the Snowflake and Snowball protocol with a weak $F$-bounded adversary ($F \geq 2$) the following properties are mutually exclusive
	\begin{itemize}
		\item The protocol ensures consensus with all but negligible probability with respect to $\beta$ (cf.\ Def.\ \ref{def:consensus}).
		\item Parties decide with less than $\pi(\beta)$ queries with all but negligible probability with respect to $\beta$, for any fixed polynomial $\pi$.
	\end{itemize}
Note that this holds even when the definition of consensus is restricted to those parties which are not influenced by the adversary.
\end{restatable}

\begin{proof}
	By Lemma \ref{lem:min_queries}, the probability that at least $c^{\beta-1}/2$ queries are required is at least $1\!-\!4/c^{2(\beta-1)}$. 
	Intuitively, this means that the Snow protocols are likely to fail if not enough queries are used, unless $\beta$ is small. 
	If, on the other hand, $\beta$ is indeed small, then there is a non-negligible chance that some party decides the minority value, thus prohibiting agreement.
	
	Formally, assume that $\beta$ is a security parameter as in Definition \ref{def:negligible}.
	Consider the threshold value \smash{$B := \frac{2}{\log c}+1$} (constant in $\beta$). 
	Consider the case that $\beta$ is small, i.e., $\beta \leq B$. 
	By Lemma \ref{lem:single_party_probability} there is a non zero probability $q_j$ (which does not depend on $\beta$) that an adversary controls at least $\alpha$ parties that have been queried by some party $j$. 
	Then the probability that the adversary can make $j$ decide on the minority value is at least $q_j^B$, which is constant in $\beta$, i.e., non-negligible.
	
	Now consider $\beta > B$. Then the probability that the protocol fails is more than $3/4$ if less than $c^{\beta-1}/2$ queries are conducted because a party did not yet decide. 
	To show the mutual exclusivity for $\beta$ above this threshold, assume that the protocol decides with all but negligible probability. 
	Then (much) more than $c^{\beta-1}/2$ queries are required by Lemma \ref{lem:min_queries}, which is larger than $\pi(\beta)$ for sufficiently large $\beta$.
\end{proof}

\section{Reconciling Security and Fast Consensus}
\label{sec:blizzard}

Theorem~\ref{thm:impossibility} shows that the Snowflake and Snowball protocols cannot achieve consensus within a polynomial number of rounds with all but negligible probability with respect to security parameter $\beta$, due to the termination condition.  We propose in Algorithm~\ref{algo:app:blizzard} a modification of the Slush protocol that we call \emph{Blizzard} and which incorporates confidence levels as a termination criterion. Importantly, Blizzard leaves the basic dynamics of the Slush protocol intact, in particular Blizzard neglects the mechanic of Snowball, where parties change their current opinion depending not only on the current but also on past queries, which is in general detrimental for the time it takes to converge to a stable consensus as shown in Section \ref{sec:snowball_dynamics}.

This modification is arguably simpler than the corresponding mechanisms in Snowball and works as follows. Each party  maintains two counters, which track the total number of queries that contained at least $\alpha$ of opinion 0 or 1, respectively (an $\alpha$-majority). A decision in favor of one opinion is made if the corresponding counter has a decisive lead over the other. The lead that is required is of the order $\bigO(\log n + \beta)$ and we show that this also corresponds to the number of rounds until consensus is established with all but negligible probability (w.r.t. $\beta$).

The idea is that within the given time frame, the network will reach a state of a stable consensus, and it will remain close to this state for a sufficiently long time such that each party can establish a lead in the counter for the opinion which is in the majority. Furthermore, unanimity is ensured because the required lead is large enough such that no party can accidentally make a ``premature'' decision, i.e., reaching the threshold even when no stable majority has been established yet. 

Note that as time progresses and given an adversary that controls at least $\alpha$ parties, there is always a small but non-zero chance that a system reverts from a state of stable consensus and even switches majorities. However, we show that once the system is in a stable consensus, it will not ``slide back'' too far within a given time frame, i.e., one opinion retains an overwhelming majority for a sufficient amount of time, even in the presence of an adversary. We start with a lemma about the probabilities to sample an $\alpha$-majority of an opinion given that one opinion has a majority.

\begin{restatable}{lemma}{lemswitchingprobabilities}
	\label{lem:switching_probabilities}
	Let $S_i \geq \frac{15n}{16}$ and $\alpha = \lceil \frac{k+1}{2}\rceil$. Then 
	$
		\P\big(Y_{ij} \geq \alpha \big) \geq p_i \text{ and } \P\big(Y_{ij} \leq k \m \alpha \big) \leq 4(1 \m p_i)^2.
	$
\end{restatable}

\begin{proof}
	Note that $1 \m p_i \leq \frac{1}{16}$. From the choice of $\alpha$, it follows that there is always an $\alpha$ majority for either 0 or 1. Thus, it suffices to show that $\P\big(Y_{ij} \leq k \m \alpha \big) \leq 4(1 \m p_i)^2 $ the other inequality follows due to $4(1 \m p_i)^2 \leq 1 \m p_i$.
	\begin{align*}
		\P\big(Y_{ij} \leq k \m \alpha \big) &  = \sum_{\ell = 0}^{k-\alpha} \binom{k}{\ell} p_i^\ell (1 \m  p_i)^{k-\ell} \leq \sum_{\ell = 0}^{k-\alpha} \binom{k}{\ell} (1 \m p_i)^{k-\ell}\\
		& \leq \sum_{\ell = 0}^{k-\alpha} \binom{k}{\ell} (1 \m  p_i)^{\alpha} =  (1\m p_i)^{\alpha} \sum_{\ell = 0}^{k-\alpha} \binom{k}{\ell}\\
		& \leq  (1\m p_i)^{\alpha} \sum_{\ell = 0}^{\lfloor k/2 \rfloor} \binom{k}{\ell} \leq  (1\m p_i)^{\alpha} 2^{k-1}   \\
		& \hspace*{-3.6mm}\stackrel{k\leq 2\alpha-1}{\leq}   (1\m p_i)^{\alpha} 2^{2\alpha-2} =   (1\m p_i)^2(1\m p_i)^{\alpha-2} 2^{2\alpha-2}\vphantom{\sum_{\ell = 0}^{\lfloor k/2 \rfloor}}\\
		& \leq   (1\m p_i)^2(\tfrac{1}{16})^{\alpha-2} 2^{2\alpha-2} = (1\m p_i)^2(\tfrac{1}{2})^{4\alpha-8} 2^{2\alpha-2}\vphantom{\sum_{\ell = 0}^{\lfloor k/2 \rfloor}}\\
		& = (1\m p_i)^2(\tfrac{1}{2})^{2\alpha-6} = 4(1\m p_i)^2(\tfrac{1}{2})^{2\alpha-4} \stackrel{\alpha \geq 2}{\leq} 4(1\m p_i)^2 \vphantom{\sum_{\ell = 0}^{\lfloor k/2 \rfloor}}
		\tag*{\qedhere}
	\end{align*}
\end{proof}

The next lemma shows that once the network is in a stable consensus state, it will very likely conserve a majority for a certain time frame.

\begin{restatable}{lemma}{lemstability}
	\label{lem:stability}
	Let $s,t \geq 1$ with $s \leq \frac{t}{2}$ and $t \leq \sqrt{n}/16$. Assume the number of parties with opinion 1 is currently $S_0 \geq n - s \sqrt n$. Then $S_i \geq n-t\cdot \sqrt{n}$ for at least \smash{$i \leq T := \min \big(\tfrac{\sqrt{n}}{32t}, \tfrac{t}{4}\big)$} rounds with probability at least \smash{$e^{-t^2}$}. This holds even in the presence of a $\sqrt{n}$-bounded adversary.
\end{restatable}

\begin{proof}
	We assume that the condition $p_i \geq 1- {t}/{\sqrt{n}}$ is true and show for how long it can be maintained starting from round $i=0$ in state $S_0$. The claim is initially true due to $s \leq \frac{t}{2}$. Since $t \leq \sqrt{n}/16$, this assumption entails $p_i \geq 15/16$, which satisfies the precondition of Lemma \ref{lem:switching_probabilities}.
	
	Let $B_{ij}$ the indicator variable that some party changes its opinion from 1 to 0 in round $i$, i.e., $B_{ij} = 1$ if that is the case, $B_{ij} = 0$ else. Then $B_i := \sum_{j = 1}^{n} B_{ij}$ is the "backslide", i.e., the number of nodes that change opinion from 1 to 0, which is a sum of independent, identically distributed Bernoulli variables. We obtain
	\[
	\E(B_i) = \sum_{j = 1}^{n} \E(B_{ij}) = n \P(X_{ij}\! =\! 0 \mid X_{i-1,j} \!=\! 1) \leq \sum_{j = 1}^{n} \P(Y_{ij} \!\geq\! \alpha) =  \sum_{j = 1}^{n} (1 \m p_i) = n(1 \m p_i)  \!\!\stackrel{\text{Lem.}\ref{lem:switching_probabilities}}{\leq}\!\! 4n\Big(\frac{t}{\sqrt{n}}\Big)^2 = 4t^2.
	\]
	We will assume that in each round $i \leq T$ the backslide is $B_i \leq 2 \E(B_i) = 8t^2$ and show that this is the case with high probability later. Considering this upper bound for the backslide, that the adversary can change the opinion of at most $\sqrt{n}$ parties in each round and exploiting the bounds for $s$ and $T$, then in round $i= T$ we obtain
	\begin{align*}
		S_T & \geq n- s\cdot \sqrt{n} - T \cdot 8t^2  - T \cdot \sqrt{n}\\
		& \geq n- \frac{t}{2} \cdot \sqrt{n} - \frac{\sqrt{n}}{32t} \cdot 8t^2 - \frac{t}{4} \cdot  \sqrt{n}\\
		& = n- \frac{t}{2} \cdot \sqrt{n} - \frac{t}{4} \cdot \sqrt{n} - \frac{t}{4} \cdot  \sqrt{n} = n - t \cdot \sqrt{n}
	\end{align*}
	In particular, this also implies $S_i \geq n - t \cdot \sqrt{n}$ thus $p_i \geq 1- t/\sqrt{n}$  for all $i \leq T$ since $B_i \leq 2 \E(B_i)$ for all $i \leq T$. It remains to compute the probability for the latter. We apply a Chernoff Bound (given in Lemma \ref{lem:chernoffbound}) and obtain $\P\big(B_i > 2\E(B_i)\big) \leq \exp\big(-\frac{4t^2}{3}\big).$ A union bound shows that this holds for all $i \leq T$ with probability
	\begin{align*}			
		\P\Big(\bigcup_{i \leq T} B_i \leq  2\E(B_i)\Big) &  =  1-\P\Big(\bigcap_{i \leq T} B_i > 2\E(B_i)\Big) 
		\geq 1- \sum_{i \leq T}  \P\big(B_i > 2\E(B_i)\big)\\
		& \geq 1- T  \cdot \exp\big(\m\tfrac{4t^2}{3}\big)
		\geq 1- \frac{t}{4} \cdot  \exp\big(\m\tfrac{4t^2}{3}\big)\\
		& \geq 1- \exp\big(\tfrac{t}{4} \big) \cdot  \exp\big(\m\tfrac{4t^2}{3}\big) = 1- \exp\big(\tfrac{t}{4} -\tfrac{4t^2}{3}\big) \stackrel{t \geq 1}{\geq} 1- e^{-t^2}.\tag*{\qedhere}
	\end{align*}	
\end{proof}

While Lemma \ref{lem:stability} captures the stability of a state of almost consensus in  a more general way, we will use it in the following form, by specifying some parameters.

\begin{restatable}{lemma}{lemstabilityspecific}
	\label{lem:stability_specific}
	Let $s \geq 1$ be a constant and assume the number of parties with opinion 1 is $S_0 \geq n - s \sqrt n$. For an arbitrary constant $\beta \geq s/2$, let $T \geq \beta$ and $T \in o(n^{1/4})$. Then $S_i \geq \frac{15n}{16}$ for at least $T$ rounds with probability at least \smash{$1-e^{-\beta^2}$}. This holds for sufficiently large $n$ even with a $\sqrt{n}$-bounded adversary.
\end{restatable}

\begin{proof}
	Let $t := 4T$, i.e., $T \leq \tfrac{t}{4}$ by definition. Since $T \in o(n^{1/4})$ we have $t \leq \sqrt n / 16$ for large enough $n$. 
	Moreover, we have $\sqrt{n}/32t \in \Omega(\!\sqrt{n}/T) \subseteq \Omega(n^{1/4})$ and therefore $T \leq \sqrt{n}/32t$ for sufficiently large $n$. Furthermore, $s \leq 2 \beta \leq 2 T \leq \tfrac{t}{2}$. This satisfies all the requirements for Lemma \ref{lem:stability} and we conclude that $S_i \geq n- t\sqrt n \geq  15n/16$ with probability at least \smash{$1-e^{-t^2} = 1-e^{-16T^2} \geq 1-e^{-\beta^2}$}.
\end{proof}

We will now use the stability properties of a network that is in a state of stable consensus (Definition \ref{def:consensus_state}) to show that Blizzard ensures consensus  (Definition \ref{def:consensus})  with all but negligible probability after $\bigO(\beta + \log n)$ rounds. Recall that Blizzard essentially corresponds to running an instance of Slush, where each node maintains counters that track how often an $\alpha$-majority was observed for opinion 0 and 1, respectively. If the difference in counters reaches a (sufficiently large) threshold $\tau$, then a decision for the opinion with the larger counter is made. In the following theorem we use a variable running time for Slush (due to the fact that as speedups up to $\log k$ over our upper bound can not be excluded, see Theorems \ref{thm:slush_lower_bound} and \ref{thm:slush_upper_bound}).

\begin{restatable}{theorem}{thmsecurityblizzard}
	\label{thm:security_blizzard}
	 Algorithm \ref{algo:app:blizzard} (Blizzard) with a threshold $\tau := 2 T_{\text{Slush}}$ ensures consensus with all but negligible probability (w.r.t.\ $\beta$) after at most $7 T_{\text{Slush}}$ rounds, assuming that $T_{\text{Slush}}$ is the number of rounds until Slush reaches a state where at least $n - \bigO(\!\sqrt{n})$ parties have the same opinion and $T_{\text{Slush}}$ is a sufficiently large multiple of $\beta$. This holds even with a $\sqrt{n}$-bounded adversary.
\end{restatable}

\begin{proof}
	Let $s \geq 1$ such that at least $S_0 \geq n - s\sqrt{n}$ parties have opinion 1 (w.l.o.g.) with all but negligible probability after at most $T_{\text{Slush}}$ rounds. Note that we have $T_{\text{Slush}} \in \bigO(\beta + \log n)$ by Theorem \ref{thm:slush_upper_bound}.
	
	Purely syntactically, we reinitialize the round counter to $i=0$ when we reach such a state $S_0$ for the first time. Starting from round $i=0$, by Lemma \ref{lem:stability_specific} we have $p_i \geq \frac{15}{16}$ for at least \smash{$T \in o(n^{1/4})$} rounds, even with a $\sqrt{n}$-bounded adversary.
	
	Let $Z_{j,0}$ and $Z_{j,1}$ be the number of times party $j$ observed an $\alpha$-majority of 0 and 1 respectively in any query made from round $i=0$ to round $i=T$. For the expectations we have.
	\begin{align*}
		\E(Z_{j,0}) & = \P\big(Y_{ij} \leq k-\alpha\big) \cdot T \stackrel{\text{Lem.}\ref{lem:switching_probabilities}}{\leq} 4(1-p_i)^2 \cdot T\leq 4 \cdot \tfrac{1}{16^2} \cdot T \leq  \tfrac{T}{16}.\\
		\E(Z_{j,1}) & = \P\big(Y_{ij} \geq \alpha\big) \cdot T \stackrel{\text{Lem.}\ref{lem:switching_probabilities}}{\geq} p_i \cdot T \geq \tfrac{15T}{16}.
	\end{align*}
	Note that the variables $Z_{j,0}$ and $Z_{j,1}$ can be seen as sums of independent Bernoulli variables, which allows us to apply the following Chernoff bounds.
	\begin{align*}
		&\P\big(Z_{j,0} \geq  (1+1) \cdot \tfrac{T}{16}\big)  \leq \exp \Big(- \tfrac{T}{3\cdot 16} \Big) \stackrel{T \geq 48\beta}{\leq} e^{-\beta}\\		
		&\P\big(Z_{j,1} \leq  (1-\tfrac{1}{3}) \cdot \tfrac{15T}{16}\big)  \leq \exp \Big(- \tfrac{15T}{9 \cdot 16} \Big) \stackrel{T \geq 48\beta/5}{\leq} e^{-\beta}.
	\end{align*}
	This implies that $Z_{j,0} <  \tfrac{T}{8}$ and $Z_{j,1} > \tfrac{5T}{8}$ with all but negligible probability given that $T \geq 48\beta$. Let $D_j(T) := Z_{j,1}-Z_{j,0}$ be the difference in the counters after $T$ rounds, then we have $D_j(T) \geq \tfrac{5T}{8}- \tfrac{T}{8} = \tfrac{T}{2}$ after $T$ rounds with all but negligible probability. 
	We can now show that the properties of consensus in Definition \ref{def:consensus} are met after $7T_{\text{Slush}} + \bigO(\beta)$ rounds with all but negligible probability. 
	
	\textbf{Termination.} We have to show that for any given party $j$ the absolute value of the difference in the counters $D_j' = cnt[1] - cnt[0]$ reaches the threshold $\tau$ eventually so that it decides (cf.\ Algorithm \ref{algo:app:blizzard}). We assume that after $T_{\text{Slush}}$ rounds we reach a stable consensus where $S_0$ parties have opinion 1, whereas the opposite case with a stable consensus for opinion 0 is analogous.
	Note that in $T_{\text{Slush}}$ rounds we have $D_j' \geq -T_{\text{Slush}}$, as in each round the counter $cnt[0]$ increases by at most one. 		
	After $T = 6T_{\text{Slush}}$ additional rounds starting from $S_0$ we obtain $D_j' \geq -T_{\text{Slush}} + D_j(T) \geq -T_{\text{Slush}} + 3T_{\text{Slush}} = 2 T_{\text{Slush}} = \tau$. Therefore, any given party decides and the algorithm terminates after $7T_{\text{Slush}}$ rounds, with all but negligible probability, given that $T \geq  48 \beta$ and \smash{$T \in o(n^{1/4})$}. The former means that we require $T_{\text{Slush}} \geq 8\beta$. The latter is fulfilled as 	$T_{\text{Slush}} \in O(\log n)$.	
	
	\textbf{Validity.} Suppose all parties propose 1 (w.l.o.g.). Then there is also a stable consensus for 1 even if we assume that initially $\sqrt{n}$ parties are flipped to 0 by an adversary. We showed above that from a stable consensus state, for any given node, the threshold $cnt[1] - cnt[0] = \tau$ will be reached with all but negligible probability.
	
	\textbf{Integrity.} Follows form the construction of the algorithm.
	
	\textbf{Agreement.} The idea is that the threshold $\tau$ is so large that no party can decide before a stable consensus is reached, and from there on every party must decide the same. Since in each round a counter of any party $j$ can increase by only one, no party can decide before $\tau = 2 T_{\text{Slush}}$ rounds have passed. Already after $T_{\text{Slush}}$ rounds we will reach a stable consensus with a $S_0$ majority for opinion 1 (w.l.o.g.). When we arrive at the state of stable consensus we have $D_j' \geq - T_{\text{Slush}}$, hence reaching the $- \tau = -2 T_{\text{Slush}}$ threshold for deciding 0 would require at least $T_{\text{Slush}}$ additional rounds. We already showed that starting from $S_0$ the change in $D_j'$ is strictly positive with all but negligible probability (provided that $T_{\text{Slush}}$ is a sufficiently large multiple of $\beta$). Consequently for any party $j$ the probability to attain the threshold $- \tau$ within the next $6T_{\text{Slush}}$ rounds is negligible.	
\end{proof}

The theorem above is given in a more general way, depending on the number of rounds until Slush reaches a stable consensus. Given that we have already shown an upper bound of $\bigO(\log n + \beta)$ for this (Corollary \ref{cor:slush_upper_bound}), we can rephrase the theorem as follows.

\begin{corollary}	
	\label{cor:security_blizzard}
	Algorithm \ref{algo:app:blizzard} (Blizzard) ensures consensus with all but negligible probability after at most $\bigO(\log n + \beta)$ rounds. This holds even with a $\sqrt{n}$-bounded adversary.
\end{corollary}

\section{Conclusion}

With the goal of improving latency in mind, we deduce two main recommendations for changes to Snow-style consensus protocols as they are deployed in the Avalanche network today. First, for a given $k$ we recommend to choose $\alpha > k/2$ as small as possible, i.e., $\alpha = \lceil \frac{k+1}{2}\rceil$, as this promises a better performance compared to $\alpha$ closer to $k$ (see Lemma \ref{lem:delta_domination} or Figure \ref{fig:delta_b}). Second, we propose to change the termination condition of the Snowflake and Snowball protocols where we observe an unfavorable trade-off between security and latency (see Corollary \ref{thm:impossibility}) to the simpler one of the Blizzard protocol (Algorithm \ref{algo:app:blizzard}). This modification will resolve the observed issue (see Theorem \ref{thm:security_blizzard} and Corollary \ref{cor:security_blizzard}). We caveat our recommendations by noting that there might be other considerations than the asymptotic performance aspects analyzed in this paper.

From a theoretical point of view, we see our results on the performance of Slush (Theorem \ref{thm:slush_lower_bound} and \ref{thm:slush_upper_bound}) as a natural continuation of the corresponding analysis of the randomized, self-stabilizing consensus protocols in the GOSSIP model where the sample size is at most 3 (such as the Median Protocol, the 2-Choices Protocol and the 3-Majority Protocol, see Definition \ref{def:randomized_consensus_protocols}). These protocols have been analyzed with respect to performance as a function of the initial number of opinions and it was shown that in ``most'' conditions they converge quite fast (cf.\ related work Section \ref{sec:related_work}). An interesting avenue of future research is the adaptation of Slush to multiple opinions, a party adopts a new opinion if it has a simple majority in a sampling of size $k$, i.e., a ``$k$-Majority protocol''. We believe that our technical work on Slush gives insights how such a $k$-Majority protocol would perform on multiple opinions.

\printbibliography

\appendix

\section{Probabilistic Concepts}
\label{sec:probabilistic_concepts}

We give a few basic definitions and principles pertaining the probabilistic security properties of some protocol (used in Definition \ref{def:negligible}) that we use throughout the paper.

\begin{definition}[Negligible Function]
	\label{def:negligible_function}
	A function $f$ is negligible if for any polynomial $\pi$ there is a constant $\lambda_0 \geq 0$, s.t., for any $\lambda \geq \lambda_0$ it is $f(\lambda) \leq \pi (\lambda)$.
\end{definition}

\begin{remark}
	We often use that for any constant $c > 0$, the function $f(\lambda) = e^{-c\cdot \lambda}$ is a negligible w.r.t.\ $\lambda$. 
\end{remark}

We usually aim for a certain security threshold given by a variable $\gamma$.

\begin{definition}[All But Negligible Probability]
	\label{def:all_but_neglible_probability}
	An event is said to occur with all but negligible probability with respect to some parameter $\lambda$ if the probability of the event not happening is a function in $\lambda$ that is negligible w.r.t. $\lambda$.
\end{definition}

In the literature, randomized consensus protocols are often shown to be successful \emph{with high probability}, which expresses the probability of failure as a function that decreasing inversely with the input size $n$ of the problem (here $n$ is the number of parties). This is often quite convenient as it eliminates any other variable from the analysis, compared to defining some fixed failure threshold.

\begin{definition}[With High Probability]
	\label{def:whp}
	An event is said to hold with high probability (w.h.p.),  if there exists a constant $c\geq1$ such that the event occurs with probability at least $1-n^{-c}$  for sufficiently large $n$.
\end{definition}

The disadvantage of the notion w.h.p.\ is that it usually looks at the asymptotic behavior of a system (i.e., for large $n$), which does not provide a fixed security level for small $n$, which can be a requirement in practice.
The following lemma provides an interface between the two notions.

\begin{lemma}
	\label{lem:whp_abn}
	Let $E$ be an event such that the probability that $E$ does not occur within any interval of $T$ consecutive rounds is at most $\tfrac{1}{n}$ (which is guaranteed if $E$ occurs w.h.p., see Definition \ref{def:whp}). Then $E$ occurs with probability $1-{e^{-\lambda}}$ (all but negligible, see Definition \ref{def:all_but_neglible_probability}) after $T' = T(\tfrac{\lambda }{\ln n}\p1)$ rounds for security parameter $\lambda > 0$.
\end{lemma}

\begin{proof}
	We consider the cases $\lambda < \ln n$ and $\lambda \geq \ln n$ starting with the former. Note that $T' \geq T$, thus the probability that $E$ does not occur is at most
	\[
	\frac{1}{n} = e^{-\ln n} < e^{-\lambda}.
	\]	
	In the second case, we use that $T' \geq \tfrac{\lambda}{\ln n}T$. The probability that $E$ does not occur after $\tfrac{\lambda}{\ln n}T$ rounds is at most 
	\[\frac{1}{n^{\frac{\lambda}{\ln n}}} = \frac{1}{e^{\tfrac{\lambda\ln n}{\ln n}}} = e^{-\lambda}.\]	
\end{proof}

\begin{remark}
	\label{rem:whp_abn}
	In Lemma \ref{lem:whp_abn}, given that some event $E$ occurs w.h.p., within $T \in \bigO(\log n)$ rounds, then $T' = T(\tfrac{\lambda }{\ln n}+1) \in \bigO(\lambda + \log n)$ rounds are sufficient that $E$ occurs with all but negligible probability.
\end{remark}

\begin{lemma}[Closure of Negligible Functions]
	\label{lem:unionbound}
	Let $\lambda > 0$ and $k \leq \rho(\lambda)$ for some polynomial $\rho$. Let $E_1, \ldots ,E_k$ be outcomes (events) of some algorithm $\mathcal A$. Assume that any individual event $E_i$ takes place with all but negligible probability (Def. \ref{def:negligible}).  Then $E \coloneqq \bigcap_{i=1}^{k} E_i$ takes place with all but negligible probability.
\end{lemma}

\begin{proof}
	There is a $\lambda_0 \geq 0$ such that $\rho(\lambda) \leq e^{\lambda/2}$ for all $\lambda \geq \lambda_0$. Let $\lambda_1, \ldots , \lambda_k \in \mathbb{N}$ such that for all $i \in \{1, \ldots, k\}$ we have $\mathbb{P}(\overline{E_i}) \leq e^{-\lambda}$ for $\lambda > \lambda_i$.
	With Boole's inequality (a.k.a.\ ``Union Bound'') we have that
	\begin{align*}
		\mathbb{P}\big(\overline{E}\big) \!= \mathbb{P}\Big(\bigcup_{i=1}^{k} \overline{E_i} \Big) \leq \sum_{i=1}^{k} \mathbb{P}(\overline{E_i}) \leq \sum_{i=1}^{k} e^{-\lambda} \leq \rho(\lambda)e^{-\lambda} \leq e^{-\lambda/2} = e^{-\lambda'}
	\end{align*}
	for $\lambda' \geq 2\max(\lambda_0, \ldots ,\lambda_k)$. 
\end{proof}

\begin{remark}
	\label{rem:unionbound_n}
	Lemma \ref{lem:unionbound} implicitly shows that a series of events \emph{all} occur with all but negligible probability w.r.t. $\lambda$, given that the number of events is not too large in $\lambda$.
	In our proofs we often apply this mechanic without specifically mentioning the lemma. Sometimes we also have a number of events that scales in $n$, for this we (necessarily have to) assume that $n \leq \pi(\lambda)$ for some fixed but arbitrary polynomial $\pi$, which allows us to apply Lemma \ref{lem:unionbound} for such a number of events as well.
\end{remark}

We use the following forms of Chernoff bounds in some of our proofs.

\begin{lemma}[Chernoff Bound]
	\label{lem:chernoffbound}
	Let $X = \sum_{i=1}^n X_i$ for i.i.d.\ random variables $X_i \in \{0,1\}$ and $\mathbb{E}(X) \leq \mu_H$ and $\delta \geq 1$, then
	$$\mathbb{P}\big(X \geq (1 \!+\! \delta) \mu_H\big) \leq \exp\Big(\!-\!\frac{\delta\mu_H}{3}\Big),$$
	 Similarly, for $\mathbb{E}(X) \geq \mu_L$ and $0 \leq \delta \leq 1$ we have
	$$\mathbb{P}\big(X \leq (1 \!-\! \delta) \mu_L\big) \leq \exp\Big(\!-\!\frac{\delta^2\mu_L}{2}\Big).$$
\end{lemma}

\section{Sums over Binomial Coefficients}
\label{sec:binom}

We will summarize a few equations for sums over binomial coefficients of the form $\sum_{\ell=i}^j \binom{k}{\ell}$ that we use throughout this paper. 

\begin{remark}[Some Identities]
Starting with the definition $\binom{k}{\ell} := \frac{k!}{\ell!(k-\ell)!}$, the first observation is that $\binom{k}{\ell} = \binom{k}{k-\ell}$ by symmetry. By reordering summands we obtain $$\sum_{\ell=0}^{m} \binom{k}{\ell} = \binom{k}{0} + \dots + \binom{k}{m} = \binom{k}{k} + \dots + \binom{k}{k-m} = \binom{k}{k-m} + \dots + \binom{k}{k} = 
\sum_{\ell=k-m}^{k} \binom{k}{\ell}.$$
\end{remark}

Note that a closed form for partial sums of the form $\sum_{\ell=i}^j \binom{k}{\ell}$ is currently not known. In this article we make use a few known special cases.

\begin{remark}[Closed Forms]
	The binomial theorem implies $\sum_{\ell=0}^k \binom{k}{\ell} = 2^k$. Using the equality from the previous remark, for uneven $k$ we obtain $$\sum_{\ell=0}^{\lfloor k/2 \rfloor} \binom{k}{\ell} = \sum_{\ell=\lceil k/2 \rceil}^{k} \binom{k}{\ell} =  2^{k-1}.$$ For any $k$ and any $m < k/2$ we have at least the following estimations $$\sum_{\ell=0}^{m} \binom{k}{\ell} = \sum_{\ell=k-m}^{k} \binom{k}{\ell} \leq  2^{k-1}.$$
\end{remark}

\section{Pseudocode}
\label{apx:pseudocodes}
\label{apx:snowball_code}
\label{apx:slush_code}
\begin{algo*}
	\begin{numbertabbing}
		xxxx\=xxxx\=xxxx\=xxxx\=xxxx\=xxxx\=xxxx\kill
		\textbf{Global parameters and state}\\
		$\CN$ \` // set of parties\label{}\\
		$\var{newRound}\in\{\false,\true\}$ \` // variable indicating when to start a round\label{}\\
		$\var{decided}\in\{\false,\true\}$ \` // variable indicating when to finish the protocol\label{}\\
		$b\in\{\ 0,1, \perp\}$ \` // current estimate for decision, initially $\bot$\label{}\\
		$k\in\mathbb{N}$  \` // number of parties queried in each poll\label{}\\
		$\alpha\in\mathbb{N}$ \` // majority threshold for queries\label{}\\
		$\CS:\op{HashMap}[\CT\to\CN]$ \` // set of sampled parties to be queried\label{}\\
		$\var{votes}:\op{HashMap}[\{0,1\}\to\mathbb{N}]$ \` // number of votes for a value\label{}\\
		$\var{cnt}\in\mathbb{N}$ \` // counter for acceptance Snowflake and Snowball\label{}\\
		$\var{cnt}[0]\in\mathbb{N}$ \` // counter for acceptance Blizzard\label{}\\
		$\var{cnt}[1]\in\mathbb{N}$ \` // counter for acceptance Blizzard\label{}\\
		$\beta\in\mathbb{N}$ \` // threshold for acceptance\label{}\\
		$d:\op{HashMap}[\{0,1\}\to\mathbb{N}]$ \` // confidence value of a transaction\label{}\\
		$\var{round}\in \mathbb{N}$\` // current round\label{}\\
		$\var{maxRound}\in \mathbb{N}$\` // maximum round\label{}
	\end{numbertabbing}
	\caption{State}
	
	\label{algo:state}
\end{algo*}

\begin{algo*}
	\begin{numbertabbing}
		xxxx\=xxxx\=xxxx\=xxxx\=xxxx\=xxxx\=xxxx\kill

		\textbf{upon} $\op{propose}(b')$ \textbf{do}\label{}\\
		\>$\var{decided}\gets\false$\label{}\\
		\>$\var{newRound}\gets\true$\label{}\\
		\>$b\gets b'$\label{}\\
		\\
		\textbf{upon} $\var{newRound}\land\neg\var{decided}$ \textbf{do}\` // still not decided\label{}\\
		\>$\var{newRound}\gets\false$\label{}\\
		\>$\op{votes}[*]\gets 0$\label{}\\
		\>$\var{round}\gets \var{round}+1$\label{}\\
		\> \textbf{if} $b\neq\perp$ \textbf{then}\label{}\\
		\>\>$\CS \gets\op{sample}(\CN\setminus\{j\}, k)$\` // sample $k$ parties \label{}\\
		\>\> \op{send} message \msg{Query}{b} to all parties $k\in \CS$\label{}\\    
		\\

		\textbf{upon} $\var{votes}[b']\geq \alpha$ \textbf{do}\` // $b'=0$ or $b'=1$\label{}\\
		\> $b\gets b'$\label{}\\ 
		\> $\var{newRound}\gets\true$\label{}\\
		\\
		\textbf{upon} $n=k\land \var{votes}[0]<\alpha\land \var{votes}[1]<\alpha$ \textbf{do}\` //  no majority\label{}\\
		\>$\var{newRound}\gets\true$\label{}\\
		\\
		\textbf{upon} receiving message \msg{Query}{b'} from party $k$ \textbf{do}\label{}\\
		\>\textbf{if} $b= \bot$ \textbf{then}\label{}\\
		\>\>$\var{decided}\gets\false$\label{}\\
		\>\>$b\gets b'$\label{}\\
		\>\op{send} message \msg{Vote}{b} to party $k$\` // reply with the local value of $b$\label{}\\
		\\
		\textbf{upon} receiving message \msg{Vote}{b'} from a party $k\in\CS$ \textbf{do}\` // collect the vote $b'$\label{}\\
		\>$\var{votes}[b']\gets \var{votes}[b']+1$\label{}\\
		\\
		\textbf{upon} $\var{round}=\var{maxRound}\land \neg\var{decided}$ \textbf{do}\` // end of the protocol\label{}\\
		\>$\op{decide}(b)$\label{}\\
		\>$\var{decided}\gets\true$\label{}
	\end{numbertabbing}
	\caption{Slush (party $j$)}
	
	\label{algo:slush}
\end{algo*}

\begin{algo*}
	\begin{numbertabbing}
		xxxx\=xxxx\=xxxx\=xxxx\=xxxx\=xxxx\=xxxx\kill
		
		\textbf{upon} $\op{propose}(b')$ \textbf{do}\label{}\\
		\>$\var{decided}\gets\false$\label{}\\
		\>$\var{newRound}\gets\true$\label{}\\
		\>$b\gets b'$\label{}\\
		\\
		\textbf{upon} $\var{newRound}\land\neg\var{decided}$ \textbf{do}\` // still not decided\label{}\\
		\>$\var{newRound}\gets\false$\label{}\\
		\>$\op{votes}[*]\gets 0$\label{}\\
		\> \textbf{if} $b\neq\perp$ \textbf{then}\label{}\\
		\>\>$\CS \gets\op{sample}(\CN\setminus\{j\}, k)$\` // sample $k$ random parties\label{}\\
		\>\> \op{send} message \msg{Query}{b} to all parties $k\in \CS$\label{}\\    
		\\
		\textbf{upon} $\var{votes}[b']\geq \alpha$ \textbf{do}\` // $b'=0$ or $b'=1$\label{l:queryend}\\
		\>\textbf{if} $b= b'$ \textbf{then}\` // majority for our proposal\label{}\\
		\>\> $\var{cnt}\gets \var{cnt}+1$\label{}\\
		\> \textbf{else}\label{}\\
		\>\> $b\gets b'$\label{}\\
		\>\> $\var{cnt}\gets 1$\label{}\\ 
		\> $\var{newRound}\gets\true$\label{}\\
		\\

		\textbf{upon} $n=k\land \var{votes}[0]<\alpha\land \var{votes}[1]<\alpha$ \textbf{do}\` // no majority\label{}\\
		\> $\var{cnt}\gets 0$\label{}\\
		\>$\var{newRound}\gets\true$\label{}\\
		\\
		\textbf{upon} receiving message \msg{Query}{b'} from party $k$ \textbf{do}\label{}\\
		\>\textbf{if} $b= \bot$ \textbf{then}\label{}\\
		\>\>$\var{decided}\gets\false$\label{}\\
		\>\>$b\gets b'$\label{}\\
		\>\op{send} message \msg{Vote}{b} to party $k$\` // reply with the local value of $b$\label{}\\
		\\
		\textbf{upon} receiving message \msg{Vote}{b'} from a party $k\in\CS$ \textbf{do}\` // collect the vote $b'$\label{}\\
		\>$\var{votes}[b']\gets \var{votes}[b']+1$\label{}\\
		\\
		\textbf{upon} $\var{cnt}=\beta \land \neg\var{decided}$ \textbf{do}\` // there is enough confidence for $B$\label{}\\
		\>$\op{decide}(b)$\label{}\\
		\>$\var{decided}\gets\true$\label{}
	\end{numbertabbing}
	\caption{Snowflake (party $j$)}
	
	\label{algo:snowflake}
\end{algo*}

\begin{algo*}
	\begin{numbertabbing}
		xxxx\=xxxx\=xxxx\=xxxx\=xxxx\=xxxx\=xxxx\kill
		
		\textbf{upon} $\op{propose}(b')$ \textbf{do}\label{}\\
		\>$\var{decided}\gets\false$\label{}\\
		\>$\var{newRound}\gets\true$\label{}\\
		\>$b\gets b'$\label{}\\
		\\
		\textbf{upon} $\var{newRound}\land\neg\var{decided}$ \textbf{do}\` // still not decided\label{}\\
		\>$\var{newRound}\gets\false$\label{}\\
		\>$\op{votes}[*]\gets 0$\label{}\\
		\> \textbf{if} $b\neq\perp$ \textbf{then}\label{}\\
		\>\>$\CS \gets\op{sample}(\CN\setminus\{j\}, k)$\` // sample $k$ random parties\label{}\\
		\>\> \op{send} message \msg{Query}{b} to all parties $k\in \CS$\label{}\\    
		\\
		\textbf{upon} $\var{votes}[b']\geq \alpha$ \textbf{do}\` // $b'=0$ or $b'=1$\label{l:queryend}\\
		\>$d[b']\gets d[b']+1$\label{}\\
		\>\textbf{if} $b= b'$ \textbf{then}\` // majority  for our proposal\label{}\\
		\>\> $\var{cnt}\gets \var{cnt}+1$\label{}\\
		\> \textbf{else}\label{}\\
		\>\>\textbf{if} $d[b']>d[b]$ \textbf{then}\label{}\\
		\>\>\> $b\gets b'$\label{}\\
		\>\> $\var{cnt}\gets 1$\label{}\\ 
		\> $\var{newRound}\gets\true$\label{}\\
		\\
		\textbf{upon} $n=k\land \var{votes}[0]<\alpha\land \var{votes}[1]<\alpha$ \textbf{do}\` // no majority\label{}\\
		\> $\var{cnt}\gets 0$\label{}\\
		\>$\var{newRound}\gets\true$\label{}\\
		\\
		\textbf{upon} receiving message \msg{Query}{b'} from party $k$ \textbf{do}\label{}\\
		\>\textbf{if} $b= \bot$ \textbf{then}\label{}\\
		\>\>$\var{decided}\gets\false$\label{}\\
		\>\>$b\gets b'$\label{}\\
		\>\op{send} message \msg{Vote}{b} to party $k$\` // reply with the local value of $b$\label{}\\
		\\
		\textbf{upon} receiving message \msg{Vote}{b'} from a party $k\in\CS$ \textbf{do}\` // collect the vote $b'$\label{}\\
		\>$\var{votes}[b']\gets \var{votes}[b']+1$\label{}\\
		\\
		\textbf{upon} $\var{cnt}=\beta \land \neg\var{decided}$ \textbf{do}\` // there is enough confidence for $B$\label{}\\
		\>$\op{decide}(b)$\label{}\\
		\>$\var{decided}\gets\true$\label{}
	\end{numbertabbing}
	\caption{Snowball (party $j$)}
	
	\label{algo:snowball}
\end{algo*}

\begin{algo*}
	\begin{numbertabbing}
		xxxx\=xxxx\=xxxx\=xxxx\=xxxx\=xxxx\=xxxx\kill

		\textbf{upon} $\op{propose}(b')$ \textbf{do}\label{}\\
		\>$\var{decided}\gets\false$\label{}\\
		\>$\var{newRound}\gets\true$\label{}\\
		\>$b\gets b'$\label{}\\
		\\
		\textbf{upon} $\var{newRound}\land\neg\var{decided}$ \textbf{do}\` // not yet decided\label{}\\
		\>$\var{newRound}\gets\false$\label{}\\
		\>$\op{votes}[*]\gets 0$\label{}\\
		\> \textbf{if} $b\neq\perp$ \textbf{then}\label{}\\
		\>\>$\CS \gets\op{sample}(\CN\setminus\{j\}, k)$\` // sample $k$ parties \label{}\\
		\>\> \op{send} message \msg{Query}{b} to all parties $k\in \CS$\label{}\\    
		\\

		\textbf{upon} $\var{votes}[b']\geq \alpha$ \textbf{do}\` // $b'=0$ or $b'=1$\label{}\\
		\> $b\gets b'$\label{}\\
		\> $\op{cnt}[b]++$\label{}\\ 
		\> $\var{newRound}\gets\true$\label{}\\
		\\
		\textbf{upon} $n=k\land \var{votes}[0]<\alpha\land \var{votes}[1]<\alpha$ \textbf{do}\` //  no majority\label{}\\
		\>$\var{newRound}\gets\true$\label{}\\
		\\
		\textbf{upon} receiving message \msg{Query}{b'} from party $k$ \textbf{do}\label{}\\
		\>\textbf{if} $b= \bot$ \textbf{then}\label{}\\
		\>\>$\var{decided}\gets\false$\label{}\\
		\>\>$b\gets b'$\label{}\\
		\>\op{send} message \msg{Vote}{b} to party $k$\` // reply with the local value of $b$\label{}\\
		\\
		\textbf{upon} receiving message \msg{Vote}{b'} from a party $k\in\CS$ \textbf{do}\` // collect the vote $b'$\label{}\\
		\>$\var{votes}[b']\gets \var{votes}[b']+1$\label{}\\
		\\
		\textbf{upon} $\op{cnt}[1]-\op{cnt}[0]= \tau \land \neg\op{decided}$  \textbf{do}\`// threshold $\tau \in \bigO(\log n + \beta)$, see Thm.\ \ref{thm:security_blizzard} and Cor.\ \ref{cor:security_blizzard}\label{}\\
		\> $\op{decide}(1)$\label{}\\
		\\
		\textbf{upon} $\op{cnt}[0]-\op{cnt}[1]=\tau \land \neg\op{decided}$ \\
		\> $\op{decide}(0)$\label{}

	\end{numbertabbing}
	\caption{Blizzard (party $j$)}
	
	\label{algo:app:blizzard}
\end{algo*}

\end{document}